\documentclass[envcountsame]{llncs}
\usepackage[utf8]{inputenc}
\usepackage{booktabs}
\usepackage[T1]{fontenc}
\usepackage[english]{babel}
\usepackage{microtype,environ,wrapfig,enumitem}
\usepackage{amsmath,amssymb,amsfonts,times,cite,xcolor,xspace,array}
\usepackage{stmaryrd,mathtools,bm,thmtools,thm-restate,fancyvrb}
\usepackage{graphicx,float,caption,listings,bussproofs,interval}
\def\orcidID#1{\smash{\href{http://orcid.org/#1}{\protect\raisebox{-1.25pt}{\protect\includegraphics{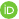}}}}}
\intervalconfig{soft open fences}
\usepackage[
  pdftex,
  colorlinks = true,
  linkcolor = blue,
  citecolor = blue,
  urlcolor = blue,
  bookmarks,
  breaklinks = true
]{hyperref}
\usepackage{tabularx,cleveref}
\usepackage[scale=.8]{noto-mono}

\input{main.sty}

\usepackage[firstpage]{draftwatermark}
\SetWatermarkText{\hspace*{10.50cm}\raisebox{13.8cm}{\includegraphics{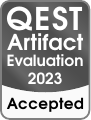}}}
\SetWatermarkAngle{0}

\title{Symbolic Semantics for Probabilistic Programs}

\author{Erik Voogd \inst{1}\orcidID{0009-0007-9712-3224} \and Einar
  Broch Johnsen \inst{1}\orcidID{0000-0001-5382-3949} \and Alexandra
  Silva \inst{2}\orcidID{0000-0001-5014-9784} \and\\ Zachary J. Susag
  \inst{2}\orcidID{0000-0002-3981-9529} \and Andrzej Wąsowski
  \inst{3}\orcidID{0000-0003-0532-2685}}

\institute{University of Oslo, Oslo, Norway \and
  Cornell University, Ithaca, New York, USA \and
  IT University of Copenhagen, Copenhagen, Denmark}

\begin{document}
\maketitle

\begin{abstract}  We present a new symbolic execution semantics of probabilistic programs
  that include observe statements
  and sampling from continuous distributions.
  Building on Kozen's seminal work, this symbolic semantics consists of a countable collection of
  measurable functions, along with a partition of the state
  space. We use the new semantics to provide a full correctness proof of symbolic execution for probabilistic programs. We also 
   implement this semantics in the
  tool \toolname, and illustrate its use on examples.
\end{abstract}

\setcounter{footnote}{0}

\section{Introduction}
Probabilistic programming languages are designed to make probabilistic
computations easier to express for a broader scientific community.
They can be used to model behaviour based on data that carries
uncertainty or randomness, as found in, e.g.,
robot\-ics~\cite{thrun2002probabilistic}, machine learning
\cite{murphy2022probabilistic,ghahramani2015probabilistic}, statistics
\cite{bayesiandataanalysis1995}, and cryptography
\cite{goldwasser1984probabilistic}.  Besides traditional programming
constructs, a key aspect of a probabilistic language is the ability to
\emph{sample} random values, in order to represent the uncertainty
that occurs in the real world.  It is essential that these programming
languages have a rigorous foundation, such that correctness and safety
properties can be guaranteed when designing and implementing tools for
probabilistic program analysis, optimization, and compilation.

Probabilistic semantics was first studied
by Kozen~\cite{kozen1979semantics}, using measure theory to relate
 denotational and operational semantics of imperative
probabilistic programs.  The \emph{denotational}
semantics represents states as probability measures over program
variables and programs as measure transformers. It enables one
to effectively reason about programs as a whole. The
\emph{operational} semantics, on the other hand, is a computation model
describing step-by-step computation in terms of a measurable function. A
correctness theorem for the operational semantics states that
the produced sets of outputs
are distributed correctly according to the measure transformer of the
denotational semantics.

To enable more detailed reasoning about probabilistic programs, we introduce
a new semantics, inspired by symbolic execution techniques.
For non-probabilistic programs, symbolic
execution~\cite{king76cacm,baldoni18compsurv} has been very successful
in program analysis techniques such as debugging, test
generation, and verification (e.g.,
\cite{deboer2021symbolic,godefroid05pldi,cadar06ccs,cadar08osdi,cadar11icse,cadar13cacm,degouw15cav,hentschel19sttt}).
Its appeal arises from the fact that one symbolic execution is an
abstraction of possibly infinitely many executions in the concrete
state space, that all share a single execution path through the
program.  Symbolic execution interprets program variables
symbolically, such that
assignments update a symbolic substitution and conditional statements
produce so-called path conditions that must be true for the program's
input variables to execute each path.
Thus, symbolic execution generates pairs consisting of a symbolic
substitution and a path condition.  If the Boolean path
condition holds in some initial state, then running
the program is the same as applying the corresponding substitution
to the initial state.

We use the new semantics to provide a full correctness proof of symbolic execution of probabilistic programs with respect to a denotational semantics à la Kozen.
The correctness proof is highly nontrivial: one semantics deals with the program's symbolic execution traces, the other interprets programs as measure transformers.
Our approach is to first prove a one-to-one correspondence between traces and the elements of our new symbolic semantics.
The latter is equivalent to the operational semantics first introduced by Kozen \cite{kozen1979semantics}, and hence to the denotational semantics.
The full details of this proof are included in the appendices.

We consider a language with \emph{observe} statements: in some
interpretations of Bayesian inference, a programmer can express to
have observed a value as a sample from some discrete or continuous
distribution \cite{staton16semantics}.  In others, including the
imperative language that we consider, observe statements are given the
semantics of asserting the truth of a provided Boolean formula.
Observe statements were not studied in Kozen's work, and an
operational semantics for probabilistic programming with observe
statements has not been formally defined in that setting.  For our
correctness proof, therefore, we must extend both denotational and
operational semantics with an interpretation of \emph{observe}, and
state and prove the corresponding correctness theorem anew.

Intuitively, observe statements in an operational semantics can be
modelled by \emph{rejection sampling}. That is, any execution that violates the observed formula is aborted. The probability mass will then
reside in execution traces where the observe condition holds, and the
probability of the set of traces where it does not hold will be
annihilated.  We formalize this semantics and prove its
correctness with respect to a denotational semantics that interprets
$\observe$ as a measure transformer that restricts
measures to the set corresponding with the Boolean formula $\bexp$.
This is an interpretation based on Bayes' theorem.
\looseness -1

\medskip

\noindent
\textbf{Contributions and Overview.}  We start the paper by introducing the language using an example program and discussing the challenges in the technical
development (\Cref{sec:pp}). The subsequent technical
sections of the paper present the following contributions:

\begin{itemize}[topsep=0pt]
\item An in-depth description of symbolic execution of probabilistic programs (\Cref{sec:symbex}), which supports discrete and continuous sampling as well as Bayesian inference through observing Boolean formulas of positive-measure sets. In symbolic execution, sampled values will be represented by symbolic variables.
\item As a main contribution we provide a full correctness proof of symbolic execution of probabilistic programs (\Cref{sec:correctness}) by introducing a new \emph{symbolic} semantics for imperative probabilistic programs and proving correspondence with established semantics (which we extend to support \emph{observe} statements).
\item To showcase symbolic execution of probabilistic programs, we have developed the tool \toolname that performs bounded symbolic execution according to our new symbolic semantics (\Cref{sec:tool}).  At the end of a symbolic
execution, \toolname reports the path condition, the substitution
mapping, and the set of Boolean formulas which correspond to the
program's observe statements.
We report on running \toolname on a series of examples to
show how our semantics can be applied.
\end{itemize}

\section{Probabilistic Programming}\label{sec:pp}
We introduce the language through an example and then present its grammar.
\subsection{Example Program} \label{sec:example}
\begin{wrapfigure}{r}{.308\textwidth}
\vspace*{-8.5mm}
\begin{lstlisting}[xleftmargin=-1.8mm]
  gender ~ bern(0.51);
  if (gender = 1) {
    height ~ norm(175,72);
  } else {
    height ~ norm(161,50);
  }
  observe (height >= 200);
\end{lstlisting}
\vspace*{-1.0cm}
\end{wrapfigure}

\noindent
Consider the probabilistic program on the right, modelling the gender distribution among people taller than $200$ centimeters.
In the Bernoulli sampling statement $\syntax{{\color{red}bern}(0.51)}$ for the variable $\texttt{gender}$ we assume that $51\%$ of the total population is male. Among men, height is normally distributed
with mean $175$ and variance $72$, and among women with mean $161$ and variance $50$.
The last line conditions the distribution on people taller than $200$ centimeters.

Symbolic execution has been very effective in analysis of non-probabilistic programs.
The technique builds variable substitutions by analyzing assignments, and resolves conditional branching and iteration by use of nondeterminism.
The conditionals are stored under variable substitution as a Boolean formula -- called the \emph{path condition} -- that needs to be satisfied by the initial state for the corresponding substitution to be a representation of the program.
Some problems arise when programs have discrete and continuous sample and observe statements, such as the above.

Consider for example a Bernoulli sampling statement $\x_i \sim \syntax{{\color{red}bern}(}\exp\syntax{)}$.
One could introduce additional sets of symbolic variables for sampling statements, but it is practically infeasible to do so for each possible bias $\exp$, especially if one allows parameters to be arithmetic expressions.
In our approach, we assume a finite number of \emph{primitive} distributions (that is, unparameterized) and encode \emph{parameterized} distributions using these primitives.
A Bernoulli sampling statement $\x_i \sim \syntax{{\color{red}bern}(}\exp\syntax{)}$ is then encoded by use of uniform continuous sampling from the unit interval $\interval 01$, written $\psample$, as follows:
\begin{lstlisting}
  x$_i$ ~ rnd; if (x$_i$ < $t$) {x$_i$ := 1} else {x$_i$ := 0}
\end{lstlisting}
A denotational semantics can be used to show soundness of such encodings.

For Gaussian distributed samples, the primitive distribution used is the \emph{standard} Gaussian distribution, which has mean zero and variance one.
To obtain a sample from a Gaussian distribution with mean $\exp_1$ and variance
$\exp_2$ -- both arithmetic expressions -- a standard Gaussian sample is scaled (multiplied by $\sqrt{\exp_2}$) and translated
(add $\exp_1$).

With these encodings, one trace through the example program above has final substitution $\sub$ and path condition $\pc$, given by
\[ \sub = \{\syntax{gender} \mapsto 1, \syntax{height} \mapsto \z_0\sqrt{72}+175\} 
\qquad \pc \equiv \y_0 < 0.51 \sand 1=1 \]
Here, $\y_0$ is a symbolic variable that is uniformly distributed over the interval $[0,1]$.
As it turns out, the probability of the path condition as measured by the input measure is the \emph{prior} probability $0.51$ of the trace.

The obtained Boolean formula from \emph{observe} statements in this symbolic execution is 
\[ \po \equiv \z_0 \cdot \sqrt{72} + 175 \geq 200,\] 
where $\z_0$ is a symbolic variable representing a standard Gaussian sample.
Measuring the set represented by this Boolean formula with the input measure yields exactly the \emph{likelihood} of the model we have in mind.
To keep the prior and the likelihood separate, we collect Boolean conditions under \emph{observe} statements in what we coin the \emph{path observations}.

This paper formally describes the procedure above for general probabilistic programs, and proves correctness with respect to a denotational semantics.
Building on work by Kozen~\cite{kozen1979semantics}, we choose to interpret substitutions from symbolic execution as measurable functions on the value space, and path conditions and path observations as measurable sets of values for the program variables.
These can then be compared against Kozen's denotational semantics, where program states are subprobability measures over the domain of values for the program variables, and a program $\p$ is a transformer of input to output measures
$\sem{p} \colon \ms{\R^n} \to \ms{\R^n}$.
The use of measure theory is unavoidable in order to handle continuous random variables. For example the semantics of \lstinline{height ~ norm(175,72)}, given an input $\mu \in \ms{\R}$, is the measure:
\[
\sem{\text{\lstinline{height ~ norm(175,72)}}}(\mu) : A \mapsto \gamma_{175,72}(A)
\]
Here, $\gamma_{175,72}$ denotes the Gaussian measure with expected value $175$ and variance $72$.

After defining the language formally in \Cref{sec:syntax}, we provide a detailed description of symbolic execution of probabilistic programs in \Cref{sec:symbex}.
There, towards a correctness proof, we also introduce the \emph{symbolic semantics}. 
This is a \emph{big-step} semantics for symbolic execution---this is stated in \Cref{thm:symbexsemantics}.
\Cref{sec:correctness} is aimed at proving correctness with respect to a measure-transforming semantics.
The correctness statement (\Cref{thm:mainresult}) says that the sum of probabilities of all symbolic execution paths, as measured by the input measure, expressed using the final substitutions, path conditions, and path observations, is the same as the probability under the denotational semantics.
The technical contribution is concluded with a proof of concept tool that implements symbolic execution and we perform some experiments with it (\Cref{sec:tool}).

\subsection{Language}\label{sec:syntax}
To express programs like the above example, we consider a language that contains basic imperative constructs (assignments, sequential composition, conditionals, and loops) together with constructs to manipulate random variables (sampling) and observe statements that allow conditioning on an observation defined by a Boolean condition:
\[
\begin{array}{l}
\begin{array}{rcl}
  \E\ni\exp & ::=&~q\in\Q\\
            &\mid&~\x_i\\
            &\mid&~\op(\exp_1,\dots,\exp_{\arop})\\
            \\
            \\
\end{array}\\
\begin{array}{l}
  \x_i\in\XX\\
  \relop\in\{<,\leq,=,!=,\geq,>\}\\
\end{array}
\end{array}\quad
\begin{array}{rcl}
 \B\ni\bexp & ::=&~\false\\
            &\mid&~\true\\
            &\mid&~\exp\relop\exp\\
            &\mid&~\bexp\por\bexp\\
            &\mid&~\bexp\pand\bexp\\
            &\mid&~\pneg\bexp\\\\
\end{array}\quad
\begin{array}{rcl}
 \P\ni p & ::=&~\pskipcol\\
         &\mid&~\assign\\
         &\mid&~\psamplecol\\
         &\mid&~\observecol\\
         &\mid&~\seq\p\p\\
         &\mid&~\pifcol\bexp\p\p\\
         &\mid&~\whilecol\bexp\p
\end{array}
\]
It defines \emph{expressions} $\exp \in \E$ over program variables $\x_i \in \XX$ (where $i \in \N$) and constants $q \in \Q$ using operators $\op$ of arity $\arop$.
Expressions $\exp \in \E$ are interpreted as functions $\interpret\exp : \R^n \to \R$ that compute the
value of $\exp$ given a valuation $\val : \{0,1,\ldots,n-1\} \to \R$ of the
program variables.
Formally, $\op$ is a function $\opsem:\R^{m} \to \R$ where $m=\arop$, and
$$\begin{array}{@{}l@{}}
  \interpret{q} (\val) = q, \;\quad
  \interpret{\x_i} (\val) = \val(i), \;\quad
  \interpret{\op (\exp_1, \ldots, \exp_{\arop})} ( \val )
= \opsem (\interpret{\exp_1}(\val), \ldots,
  \interpret{\exp_{\arop}}(\val))
  \end{array}$$
The grammar also defines \emph{Boolean expressions} $\bexp \in \B$ that can relate expressions and apply the standard Boolean operators.
We slightly overload notation and interpret Boolean expressions $\bexp$ as subsets $\interpret {\bexp}$ of $\R^n$ where the formula holds.
  For example, $\interpret\true = \R^n$ and if $\bexp = \exp_1 \relop \exp_2$ then
  $$\interpret\bexp = \{\val \in \R^n \setst \interpret{\exp_1}(\val) \relop \interpret{\exp_2}(\val)\}.$$
  Boolean conjunction, disjunction, and negation are respectively interpreted as set intersection, union, and complement in $\R^n$.

The Boolean expressions are used in the \syntax{if} and \syntax{while} conditions and in \syntax{observe} statements of \emph{program statements} $\p \in \P$.
Whenever some program $p \in \P$ is fixed, there is also a fixed amount of $n$ variables.

Sampling statements $\psample$ draw uniform random samples from the unit interval $\interval 01$.
For presentation purposes, we sample from only one primitive distribution in the formal grammar, and we do so at the level of statements rather than in expressions.
Expressions can be made probabilistic, however, by first sampling and then using the variable in a (Boolean) expression.
We also stress again that the language can be extended to support sampling from a multitude of other primitive distributions, both discrete and continuous, without affecting any of the results in this paper.

With the extensions described in \Cref{sec:example}, the probabilistic program presented there is in the language generated by our grammar.

\section{Symbolic Execution} \label{sec:symbex}
  The inductive rules that define a transition system implementing symbolic execution of probabilistic programs are presented in \cref{fig:symos}.
  \begin{figure}[t]
  \center\boxed{\begin{tabular}{c}
  \\
  \saxiom{asgn}{(\assign,\sub,\idx,\pc,\po)\tr(\pskipcol,\sub\update{\x_i}{\app\sub\exp},\idx,\pc,\po)}\\\\
  \saxiom{smpl}{(\psamplecol,\sub,\idx,\pc,\po)\tr(\pskipcol,\sub\update{\x_i}{\y_{\idx}},\idx+1,\pc,\po)}\\\\
  \saxiom{obs}{(\observecol,\sub,\idx,\pc,\po)\tr(\pskipcol,\sub,\idx,\pc,\po\sand\app\sub\bexp)}\\\\
  \saxiom{seq-0}{(\seq{\pskipcol}{\p},\sub,\idx,\pc,\po)\tr(\p,\sub,\idx,\pc,\po)}\\\\
  \srule{seq-n}
  {(\p,\sub,\idx,\pc,\po)\tr(\p',\sub',\idx',\pc',\po')}
  {(\seq\p\q,\sub,\idx,\pc,\po)\tr(\seq{\p'}\q,\sub',\idx',\pc',\po')}\\\\
  \saxiom{if-T}{(\pifcol\bexp{\p_1}{\p_2},\sub,\idx,\pc,\po)\tr(\p_1,\sub,\idx,\pc\sand\app\sub\bexp,\po)}\\\\
  \saxiom{if-F}{(\pifcol\bexp{\p_1}{\p_2},\sub,\idx,\pc,\po)\tr(\p_2,\sub,\idx,\pc\sand\app\sub{\pneg\bexp},\po)}\\\\
  \saxiom{iter-F}{(\whilecol\bexp\p,\sub,\idx,\pc,\po)\tr(\pskipcol,\sub,\idx,\pc\sand\app\sub{\pneg\bexp},\po)}\\\\
  \saxiom{iter-T}{(\whilecol\bexp\p,\sub,\idx,\pc,\po)\tr(\seq\p{\whilecol\bexp\p},\sub,\idx,\pc\sand\app\sub\bexp,\po)}\\\\
  \end{tabular}}
  \caption{Inductive transition rules for symbolic execution}
  \label{fig:symos}
  \end{figure}
  The \emph{symbolic states} in this transition system are quintuples consisting of the following data: \rom 1 a \emph{program} $\p$ that is to be executed; \rom 2 a \emph{substitution} $\sub$ of program variables to symbolic expressions (defined shortly), which captures past program behavior; \rom 3 a \emph{sampling index}
  \footnote{One for each primitive distribution: we use one in this paper for presentation purposes. Two sampling indices are used in the tool \toolname.}
  $\idx \in \N$; \rom 4 the \emph{path condition} as a precondition for this symbolic trace; and \rom 5 the \emph{path observation} as a means to bookkeep which states will be accepted by \emph{observe} statements throughout execution.
  Progression of the system depends only on the program syntax \rom 1.
  Below we provide a detailed description of the components \rom 2-\rom 5.
  The substitutions \rom 2 and path conditions \rom 4 follow established methods from symbolic execution \cite{deboer2021symbolic}; the notion of path observation \rom 5 is novel.

  The program $\pskip$ cannot make a transition and represents the \emph{terminated} program.
  The system is nondeterministic due to the pairs of rules for \emph{if} and \emph{while} statements.
  The system has infinite symbolic traces due to \refRule{iter-T}.
  Customarily, $\trCl$ denotes the reflexive-transitive closure of $\tr$.

  \subsection{Symbolic Substitutions} \label{sec:substitutions}
  To capture the behavior of a symbolic trace, a \emph{substitution} assigns to every program variable a \emph{symbolic expression}.
  Symbolic expressions are generated by the following grammar (recall that $\arop$ denotes the arity of $\op$):
\[\begin{array}{lclllllllllll}
 \SE \ni \sexp     & ::= & ~ q \in \Q
                  & \mid \x_i \in \XX
                  & \mid \y_\idx \in \YY
                  & \mid \op (\sexp, \ldots, \sexp_{\arop})
\end{array}\]
  $\SE$ extends $\E$ with the set $\YY = \{\y_0, \y_1, \dots\}$ as base cases, representing the random samples that may be drawn during execution.
  When the language samples from more than one primitive distribution, a set of symbolic variables $\YY_{\DD}$ is used for every primitive distribution $\DD$.
  For each such distribution $\DD$, one will need a sampling index $\idx_\DD$ in the transition system.
  We use $\ZZ = \{\z_0, \z_1, \dots\}$ for symbolic variables representing standard normal samples in our examples and in the tool.

  Thus, formally, substitutions are maps $\sub : \XX \to \SE$.
  We sometimes write $\sub(i)$ in lieu of $\sub(\x_i)$.
  The \emph{updated} substitution $\sub\update i\sexp$ denotes the substitution $\sub'$ such that $\sub'(j)=\sub(j)$ for $j\neq i$ and $\sub'(i)=\sexp$.
  Any substitution $\sub$ inductively extends to expressions $\exp \in \E$ by $\sub(\op(\exp_1, \ldots, \exp_{\arop})) := \op (\sub(\exp_1), \ldots, \sub(\exp_{\arop}))$.

  If symbolic execution is to conform with a denotational semantics, the symbolic substitutions have to be interpreted as concrete state transformers.
  For this, first, symbolic \emph{expressions} are interpreted as functions $\R^{n+\omega} \to \R$, where the first $n$ arguments are the values of the program variables, and the rest is an infinite stream of samples to be drawn.
  Here, one may use extra streams of sample spaces for any additional primitive distributions in the language.
  The mapping $\sinterpret {\cdot}: \R^{n+\omega} \to \R$ equips symbolic expressions with a formal interpretation as follows:
  $$\begin{array}{@{}l@{}}
    \sinterpret{q} : \rho \mapsto q,\;\qquad\qquad
    \sinterpret{\x_i} : \rho \mapsto \rho_i,\;\qquad\qquad
    \sinterpret{\y_k} : \rho \mapsto \rho_{n+k},\\[0.6em]
    \sinterpret{\op (\sexp_1, \ldots, \sexp_{\arop})} : \rho \mapsto \opsem (\sinterpret{\sexp_1}(\rho), \ldots, \sinterpret{\sexp_{\arop}}(\rho))\;
  \end{array}$$
  The symbols $\y_k$ thus pick the $k$-th sample available in $\R^\omega$.
  Symbolic expressions free of $\y_k$ are just program expressions, and their interpretations agree in the following way:
  \begin{restatable}[Substitution lemma for expressions]{lemma}{substitutionlemmaexp}\label{lem:substitutionlemmaexp}
    For all $\exp \in \E$, $\sinterpret\exp(\rho) = \interpret\exp (\res \rho n)$.
  \end{restatable}
  This interpretation of $\sexp\in\SE$ as a function $\R^{n+\omega} \to \R$ extends naturally to symbolic substitutions $\sub \in \SE^n$ as functions $\R^{n+\omega} \to \R^n$, by mere point-wise application after substitution.
  However, since we want to compose substitutions (their interpretations) due to sequencing, the codomain must be extended from $\R^n$ to $\R^{n+\omega}$:
  \begin{definition}[Interpretation of symbolic substitutions] \label{def:substitutioninterpretation}
    Let $\sub : \XX \to \SE$ be a substitution and $k \in \N$ a sampling index.
    The interpretation of $\sub$ at sampling index $k$, denoted $\sinterpretk \sub k$, is the function $\R^{n+\omega} \to \R^{n+\omega}$ defined by
    $$\sinterpretk\sub k : \rho \mapsto (\sinterpret {\sub(\x_0)}(\rho), \ldots, \sinterpret {\sub(\x_{n-1})}(\rho), \rho_{n+k}, \rho_{n+k+1}, \rho_{n+k+2}, \ldots) $$
  \end{definition}
  The first $n$ values of $\sinterpretk \sub k (\rho)$ are just pointwise applications of $\sinterpret {\cdot}$.
  The remaining elements of $\sinterpretk \sub k (\rho)$ are the same as those of $\rho$, but left-shifted $k$ positions.
  This formalizes the idea that $\sinterpretk \sub k$ has already drawn $k$ samples.

  \subsection{Path Conditions and Path Observations}
  The Boolean formulas for the condition in \emph{branching} and \emph{iteration} statements are aggregated under substitution to build the \emph{path condition} of a symbolic trace.
  The path condition represents the unique part of the input space that triggers the symbolic trace.
  Similarly, \emph{observed} Boolean formulas under substitution make up the \emph{path observation}, and represent the part of the input space that will lead to acceptance of \emph{observe} statements in this trace.

  Path conditions and observations are expressed as \emph{symbolic Boolean expressions}:
\[\begin{array}{lclllllllllll}
 \SB \ni \sbexp & ::=& \sfalse
                & \mid \strue
                & \mid \sexp\relop\sexp
                & \mid \sbexp\sor\sbexp
                & \mid \sbexp\sand\sbexp
                & \mid \sneg\sbexp
\end{array}
\]
  Subtitutions $\sub : \XX \to \SE$ extend through $\E \to \SE$ further to Booleans, as in $\B \to \SB$, in a completely trivial way.
  For example, $\sub (\exp_1 \relop \exp_2 \pand \true) = \sub(\exp_1) \relop \sub(\exp_2) \sand \strue$.

  Path conditions and path observations -- their symbolic Boolean expressions -- are interpreted as subsets of $\R^{n+\omega}$ where the formula is satisfied.
  The mapping $\sinterpret{\cdot}$ equips symbolic Boolean expressions with a formal interpretation as follows:
\begin{flalign*}
& \sinterpret{\sfalse} = \emptyset,
&& \sinterpret{\sbexp_1 \sor \sbexp_2} = \sinterpret{\sbexp_1} \cup \sinterpret{\sbexp_2},\\
& \sinterpret{\strue} = \R^{n+\omega},
&& \sinterpret{\sbexp_1 \sand \sbexp_2} = \sinterpret{\sbexp_1} \cap \sinterpret{\sbexp_2}, \hfill \\
& \sinterpret{\sneg \sbexp} = \R^{n+\omega} \setminus \sinterpret{\sbexp},
&& \sinterpret{\sexp_1\relop\sexp_2} = \{\rho \in \R^{n+\omega} \setst
\sinterpret{\sexp_1}(\rho) \relop \sinterpret{\sexp_2}(\rho) \}.
\end{flalign*}
  Here we overload notation of $\sinterpret{\cdot}$ to use it for both expressions and Boolean expressions.

  \subsection{Final Configurations}
  Let $\initsub$ denote the \emph{initial} substitution $\{\x_i \mapsto \x_i\}_{\x_i \in \XX}$.
  We call the quadruple $(\sub,\idx,\pc,\po)$ the (symbolic) \emph{configuration} of a state $(\p,\sub,\idx,\pc,\po)$ and $(\initsub,0,\strue,\strue)$ is the \emph{initial configuration}.
  The configurations we are mostly interested in result from a finite symbolic execution trace starting from the initial configuration:
  $$ \Gamma_\p := \{ (\sub,\idx,\pc,\po) \setst (\p,\initsub,0,\strue,\strue) \trCl (\pskip,\sub,\idx,\pc,\po) \} $$
  is called the set of \emph{final configurations} of $\p$.

  For a program $\p \in \P$ and a configuration $(\sub,\idx,\pc,\po)\in\Gamma_\p$, $\p$ transforms inputs $\rho\in\sinterpret {\pc} \cap \sinterpret{\po}$ to the output $\sinterpretk\sub\idx(\rho)$.
  That is, $\p$ behaves like $\sinterpretk\sub\idx$ on $\sinterpret\pc\cap\sinterpret\po$.
  Execution of $\p$ on inputs from $\sinterpret{\pc} \setminus \sinterpret {\po}$ leads to an unsatisfied \emph{observe} statement.
  \begin{example}
    Consider the example program in \Cref{sec:example} (call it $\p$) that models the gender distribution among people taller than 200 centimeters.
    Due to Bernoulli sampling, it contains an additional \emph{if} statement hidden in the encoding.
    The program has thus \emph{four} symbolic traces; their respective final configurations $\cfg_1,\cfg_2,\cfg_3,\cfg_4 \in \Gamma_\p$ are:
    \begin{figure}[H]
      \renewcommand*\arraystretch{1.3}
      \renewcommand*\tabcolsep{1.5mm}
      \centering
    \begin{tabular}{c|l|c|c|l|l}
      \textbf{Final} &
      \textbf{Substitution } $\sub$ & $\idx_{\y}$ & $\idx_{\z}$ &
      \textbf{Path condition } $\pc$ &
      \textbf{Path observation } $\po$ \\\hline
      $\cfg_1$ & $\{g \mapsto 1, h \mapsto \z_0\sqrt{72}+175\}$ & $1$ & $1$ &
      $\y_0 < 0.51 \sand 1=1$ &
      $\z_0\sqrt{72}+175 \geq 200$ \\
      $\cfg_2$ & $\{g \mapsto 1, h \mapsto \z_0\sqrt{50}+161\}$ & $1$ & $1$ &
      $\y_0 < 0.51 \sand 1\neq 1$ &
      $\z_0\sqrt{50}+161 \geq 200$ \\
      $\cfg_3$ & $\{g \mapsto 0, h \mapsto \z_0\sqrt{72}+175\}$ & $1$ & $1$ &
      $\y_0 \geq 0.51 \sand 0=1$ &
      $\z_0\sqrt{72}+175 \geq 200$ \\
      $\cfg_4$ & $\{g \mapsto 0, h \mapsto \z_0\sqrt{50}+161\}$ & $1$ & $1$ &
      $\y_0 \geq 0.51 \sand 0\neq 1$ &
      $\z_0\sqrt{50}+161 \geq 200$
    \end{tabular}
  \end{figure}

  \noindent
  The final configurations $\cfg_2$ and $\cfg_3$ have unsatisfiable path conditions.
  Path conditions represent the \emph{priors} in the model and path observations represent \emph{likelihoods}.
  \end{example}

\subsection{Symbolic Semantics}
Now we introduce symbolic semantics for probabilistic programs with the main purpose of proving correctness of symbolic execution, which we just described.
The new semantics is a set of functions, and can in fact be considered a denotational semantics for symbolic execution of probabilistic programs.
Defined directly on the syntax of programs, the semantics consists of the interpretations of all final symbolic configurations---this is \Cref{thm:symbexsemantics} below.

For a triple $(F,B,\OO)$ in the definition below, $F$ is the final substitution of a trace, $B$ is the path condition, and $\OO$ is the path observation.
Recall that $\BB(X)$ is the Borel $\sigma$-algebra of $X$ and let $\BB(X,Y)$ denote the space of Borel measurable functions $X \to Y$.
Then, for programs $\p \in \P$ in $n$ variables, the sets
  $$\FF_{\p} \subset \BB(\R^{n+\omega},\R^{n+\omega}) \times \BB(\R^{n+\omega}) \times \BB(\R^{n+\omega})$$
are defined inductively on the structure of $\p$ as follows:
  \begin{itemize}
    \item For \emph{inaction} $\pskip$, the state remains unaltered and there is no restriction on the path condition or the path observation:
      \[ \FF_{\pskip} :=  \{(\rho\mapsto\rho,\statespace,\statespace)\} \]
    \item An \emph{assignment} has no restriction on the precondition, but the state is updated according to the assignment:
      \[ \FF_{\assign} := \{(\rho\mapsto\rho\update i{\interpret\exp(\res\rho n)},\statespace,\statespace) \} \]
      Only the first $n$ values $\res\rho n$ of the state $\rho$ are needed to evaluate $\exp$, and $\rho\update ia$ denotes the state $\rho'$ where $\rho'(j)=\rho(i)$ if $j\neq i$ and $\rho'(i)=a$.
    \item When \emph{sampling} we also merely perform an appropriate state update:
      \[ \FF_{\psample} := \{(\rho\mapsto\sample i (\rho),\statespace,\statespace)\} \]
      Here, $\sample i(\rho)$ is an updated stream $\rho'$ that has drawn one sample:
      $$ \sample i (\rho) = (\rho_0, \dots, \rho_{i-1}, \bm{\rho_n}, \rho_{i+1}, \ldots, \rho_{n-1}, ~{\bm{\rho_{n+1}}}, \bm{\rho_{n+2}}, \bm{\ldots} ) $$
    \item \emph{Observing} a Boolean formula only updates the path observation:
      \[ \FF_{\observe} := \{(\rho\mapsto\rho,\statespace,\interpret{\bexp}\times\R^\omega)\} \]
    \item When \emph{sequencing} two programs $\p_1$ and $\p_2$, range over all pairs of executions $(F_1,B_1,\OO_1) \in \FF_{\p_1}$ and $(F_2,B_2,\OO_2) \in \FF_{\p_2}$ and compose them.
      The first path condition $B_1$ should be satisfied and, after executing the first component $F_1$, the second path condition $B_2$ should be satisfied (and sim. for the path observation):
      \[ \FF_{\seq{\p_1}{\p_2}} := (F_2\circ F_1,B_1\cap F^{-1}_1[B_2],\OO_1\cap F_1^{-1}[\OO_2])
    \setst (F_i,B_i,\OO_i) \in \FF_{\p_i}, i=1,2 \} \]
    \item The two branches of an \emph{if} statement are put together in a binary union of sets.
      The path conditions are updated accordingly ($\compl -$ denotes complement):
      \[ \begin{array}{rl}
        \FF_{\pif\bexp\p\q} :=
             & \{(F,B\cap(\interpret\bexp\phantom{\compl{}}\times\R^\omega),\OO)\setst(F,B,\OO)\in\FF_\p\} \\
        \cup & \{(F,B\cap(\bexpcset),\OO)\setst(F,B,\OO)\in\FF_\q\}
    \end{array}\]
    \item In a \emph{while} statement, the union is over every possible number of iterations $m$.
      For $m=0$, the behavior is that of $\pskip$ and the precondition is the negation of the Boolean formula.
      Every next number $m+1$ of loop iterations takes all possible executions of $m$ iterations, pre-composes all possible additional iterations, and updates the preconditions accordingly:
      \[ \FF_{\while \bexp\p} := \bigcup_{m=0}^\infty \mathbb F_{\bexp,\p}^m\{ (\val \mapsto \val, \compl{\interpret\bexp}, \statespace) \}, \]
      where $\mathbb F_{\bexp,\p}^m$ denotes $m$ applications of the mapping $\mathbb F_{\bexp,\p}$ from $\BB(\R^{n+\omega},\R^{n+\omega}) \times \BB(\R^{n+\omega}) \times \BB(\R^{n+\omega})$ to itself that pre-composes an additional iteration of the loop:
    $$ \begin{array}{rl}
      \mathbb F_{\bexp,\q} : \mathcal F \mapsto \{  & (F \circ F_\q,(\bexpset) \cap B_\q \cap F_\q^{-1}[B],\OO_\q \cap F_\q^{-1}[\OO]) \\
                                                      & \setst (F,B,\OO) \in \mathcal F, (F_\q,B_\q,\OO_\q) \in \mathcal F_{\q} \}
      \end{array}$$
  \end{itemize}
  The sets $\FF_\p$ form a big-step semantics for symbolic execution of probabilistic programs:
  \begin{restatable}{theorem}{symbexsemantics}\label{thm:symbexsemantics}
  For any program $\p \in \P$, there is a one-to-one correspondence between final configurations $(\sub,\idx,\pc,\po) \in \Gamma_\p$ and triples $(F,B,\OO) \in \FF_\p$ such that $F = \sinterpretk\sub\idx$, $B = \sinterpret{\pc}$, and $\OO = \sinterpret {\po}$.
  \end{restatable}
  \noindent
  In this correspondence, we consider all final configurations $(\sub,\idx,\pc,\po)\in\Gamma_\p$ with unsatisfiable path condition, i.e., $\sinterpret\pc = \emptyset$, equivalent. 
  Similarly, all triples $(F,B,\OO)$ for which $B=\emptyset$ are considered equivalent.
  The proof is in \Cref{app:symbexsemanticsproof}.

  \section{Correctness} \label{sec:correctness}
The symbolic execution engine described in the previous section is now proven correct with respect to a denotational semantics.
Following Kozen \cite{kozen1979semantics}, probabilistic programs are mappings of measures on the Borel measurable space of $\R^n$, where $n$ is the number of program variables.
Observe statements were not considered in his work. 
They \emph{have} been studied \cite{vangael2013} in this context of measure transformer semantics, but in the absence of unbounded loops.

To be fully precise, we need to recall some definitions from measure theory.
A \emph{measurable space} $(X,\Sigma)$ is a set $X$ equipped with a $\sigma$-\emph{algebra} $\Sigma$, i.e., a set $\Sigma \subseteq \powerset {X}$ that \rom{1} contains the emptyset, \rom{2} is closed under complements in $X$, and \rom{3} is closed under countable union.
Elements of $\Sigma$ are called \emph{measurable sets}.
The \emph{Borel} $\sigma$-algebra $\BB(X)$ is the one generated by the open sets of $X$.
Whenever we say that functions or sets are measurable, we mean with respect to the Borel $\sigma$-algebras.
A \emph{(sub)probability measure} on $(X,\Sigma)$ is a function $\mu : \Sigma \to \interval 01$ such that $\mu (\emptyset) = 0$ and
$\mu (\bigcup_{i \in I} A_i) = \sum_{i \in I} \mu(A_i)$ for any
countable disjoint family of sets $(A_i)_{i\in I} \subseteq \Sigma$.
We let $\ms{X}$ denote the set of measures on the Borel $\sigma$-algebra of $X$; this makes $\ms{\R^n}$ the state space in the denotational semantics.

The denotational semantics of a program
$\p$ is a function
$\sem\p : \ms{\R^n} \to \ms{\R^n}$, pushing a measure
corresponding to the current program state forward, according to the
statement being interpreted. The inductive definition is as follows:
\begin{itemize}[topsep=3pt]
\item \emph{Skip} does not change the given measure:
  $\sem{\pskip}(\mu)=\mu$.
\item For an \emph{assignment}, let $\asgn^i_\exp$ be the function that updates the $i$-th value appropriately:
$$ \asgn^i_{\exp} : \R^n \to \R^n, \quad (x_1,\ldots,x_n) \mapsto (x_1, \ldots, x_{i-1}, \interpret{\exp} (x_1,\ldots,x_n), x_{i+1}, \ldots, x_n ) $$
This function is measurable, so its preimages can be measured by $\mu$.
Thus, the semantics of assignments as \emph{pushforward} measures
$ \sem{\assign}(\mu) := \pushforward{\mu}{(\asgn^i_{\exp})} $
is well-defined.
Intuitively, the measure $\sem{\assign}(\mu)$ measures with $\mu$ the set of values in $\R^n$ that lead to appropriately updated values for the $i$-th variable.
\item \emph{Sampling} updates the measure component of the
  corresponding variable with the distribution measure $\lambda$ to be sampled from:
  $$ \sem{\psample}(\mu) : A_1 \times \dots \times A_n \mapsto \mu(A_1 \times \dots \times A_{i-1} \times \R \times A_{i+1} \times \dots \times A_n) \cdot \lambda (A_i).$$
  For $\lambda$, we use the Lebesgue measure on the unit interval for uniform continuous sampling.
  Other measures can be used for other primitive sampling statements.
  The measure $\sem{\psample}(\mu)$ here is defined on the rectangles of $\R^n$, and by
Carathéodory's extension theorem, defines a unique measure on all of
$\R^n$.
\item When \emph{observing}, we restrict the measure to the observed
  measurable set.  For measurable $B\subseteq \R^n$, let
  $\restrictMeasure B (\mu)$ denote the subprobability measure
  $A \mapsto \mu(A \cap B)$.  Then
$ \sem{\observe} (\mu) := \restrictMeasure{\interpret{\bexp}}(\mu).$
Note that this measure is not normalized.
\item \emph{Sequencing} is just composition:
  $ \sem{\seq \p\q} = \sem{\q} \circ \sem\p$.
\item \emph{Branching} restricts the measure of one branch (resp. the other) to the
  measurable subset where the condition is true (resp. false).  Formally,
$$ \sem{\pif {\bexp} {\p_1} {\p_2}}(\mu) = (\sem{\p_1} \circ \restrictMeasure{\interpret{b}})(\mu) + (\sem{\p_2} \circ \restrictMeasure{\interpret{\pneg b}})(\mu) $$
This sum of measures is
 setwise.  The measure is first restricted by
$\restrictMeasure {\interpret{\bexp}}$ (or
$\restrictMeasure{\interpret{\pneg \bexp}}$) to a
subprobability measure over the part of the state space where the
condition is true (or false) and then passed on to be transformed by $\sem{\p_1}$ (or $\sem{\p_2}$).
\item Interpreting \emph{iterations} as repeated unfolding of if
  statements, the infinite sum
$$ \sem{\while {\bexp} \p}(\mu) = \Big( \sum_{m=0}^{\infty} \restrictMeasure {\interpret{\pneg \bexp}} \circ (\sem\p \circ \restrictMeasure {\interpret{\bexp}})^m \Big) (\mu)$$
describes the semantics of while loops.
\end{itemize}

\noindent
Now $\sem\p$ is a positive
operator of norm at most one for all $\p$ defined by the
grammar in \Cref{sec:syntax}.
This means that any subprobability measure is mapped to a
subprobability measure \cite{barthe2020foundations}.
Without while
and observe statements, this norm would be exactly one.  Intuitively,
this is because programs would then always terminate, and no
probability mass would ever be lost either by diverging while loops or
by conditional probability.
The measure semantics is not normalized.

  Recall that $\Gamma_\p$ is the set of final configurations of symbolic execution of a program $\p$.
  \begin{restatable}[Correctness of Symbolic Execution of Probabilistic Programs]{theorem}{mainresult}\label{thm:mainresult}
    Let $\p \in \P$ be a program, $\mu \in \ms{\R^n}$ a distribution measure over the input variables, and $A \subseteq \R^n$ a measurable set.
    Then
    $$ \sem \p (\mu) (A) = \sum_{(\sub,\idx,\pc,\po) \in \Gamma_\p} {(\mu\otimes\infleb)} \big(\sinterpretk\sub\idx^{-1}[A \times \R^\omega] \cap \sinterpret\pc \cap \sinterpret\po \big) $$
  \end{restatable}
  \noindent A friendly reminder that the measure $\lambda$ was some chosen measure implicitly used in the denotational semantics; $\infleb$ denotes the product measure of infinitely many copies of it.
  \begin{proof}[Sketch]
  For every program $\p$, there is a function $f_\p$ such that
  \begin{itemize}
    \item $(\mu\otimes\infleb) \Big(\f^{-1}_{\p}[A\times\R^{\omega}]\Big) = \sem \p (\mu) (A)$, and
    \item $(\mu\otimes\infleb) \Big(\f^{-1}_{\p}[A\times\R^{\omega}]\Big) = \sum_{(F,B,\OO) \in \FF_\p} (\mu\otimes\infleb)(F^{-1}[A\times\R^{\omega}] \cap B \cap \OO)  $.
  \end{itemize}
  The function $f_\p$ is basically the operational semantics defined in \cite{kozen1979semantics}, but extended to an observe construct that rejects unsatisfied observed formulas.
  The theorem is proven by chaining these two equalities and applying \Cref{thm:symbexsemantics}.
  See \Cref{app:mainresultproof} for the full proof.
  \end{proof}

\section{Implementation and Experiments} \label{sec:tool}

We have developed \toolname as a prototype implementation of the
symbolic execution technique presented in
\Cref{sec:symbex}\footnote{Source code for \toolname and all
  experiments are available on Zenodo~\cite{symProbArtifact2023}.}.
\toolname~takes as input a probabilistic program written in the
language described in \Cref{sec:syntax} (up to some natural
imperative-style syntactic additions such as keywords for \emph{else}
and delimiter use of parentheses and braces). \toolname~performs
\emph{bounded} symbolic execution, meaning that all while loops are
unrolled a finite number of times, to ensure termination. Finite loops
are fully unrolled while all other loops are unrolled a configurable,
yet fixed, number of times.  \toolname~reports the final
configuration: the final substitution $\sigma$, the amount of samples,
the path condition $\pc$, and the path observation $\po$.  All
symbolic configurations reported by \toolname are expressed as
uninterpreted symbolic expressions.  \toolname~is written in Rust in
around 2,000 lines of code and uses Z3~\cite{demoura2008} for real
numbers to determine branch satisfiability.

\begin{table*}[t]
  \centering
    \begin{tabularx}{.9\textwidth}{@{}l RRRCR @{}}
        \toprule
        &\multicolumn{2}{c}{Number of Paths} \\ \cmidrule{2-3}
        \textbf{Case Study} & \textbf{Actual} & \textbf{Discarded} & \textbf{Samples} & \textbf{Lines} & \textbf{Time (sec.)}\\ \midrule
        BurglarAlarm & 4 & 12 & 4 & 26 & 0.31\\
        DieCond & 20 & 0 & 20 & 17 & 2.15 \\
        Grass & 28 & 36 & 6 & 21 & 0.80 \\
        MurderMystery & 2 & 2 & 2 & 12 & 0.06 \\
        NeighborAge & 4 & 4 & 4 & 14 & 0.10 \\
        NeighborBothBias & 7 & 5 & 4 & 19 & 0.24 \\
        NoisyOr & 256 & {--} & 8 & 37 & 2.33 \\
        Piranha & 3 & 1 & 2 & 12 & 0.07 \\
        Random Z2 Walk (1) & 16 & {--} & 4 & 14 & 0.19 \\
        Random Z2 Walk (2) & 52 & {--} & 6 & 14 & 0.60 \\
        Random Z2 Walk (4) & 712 & {--} & 10 & 14 & 8.37 \\
        Random Z2 Walk (8) & 159,436 & {--} & 18 & 14 & 2167.90 \\
        SecuritySynthesis & 1 & 0 & 8 & 14 & 0.04 \\
        TrueSkillFigure9 & 1 & 0 & 9 & 20 & 0.03 \\
        TwoCoins & 3 & 1 & 2 & 6 & 0.05 \\
        \bottomrule
      \end{tabularx}
      \medskip
    \caption{Performance metrics for \toolname on a series of case studies. For the ``Random Z2 Walk ($i$)'' cases, $i$ refers to the number of times the main \texttt{while} loop was unrolled.}
    \label{tab:casestudies}
    \vspace{-6mm}
 \end{table*}

 We have executed
 \toolname on a series of examples sourced from
 PSI~\cite{gehr2016psi}, R2~\cite{nori2014}, ``Fun'' programs from
 Infer.NET~\cite{InferNET18}, and Barthe et al.~\cite{barthe2020foundations}.
 We summarize our findings
 in Table~\ref{tab:casestudies}. All the experiments were done on a
 machine with 3.3GHz Intel Core i7-5820K and 32 GB of RAM, running
 Linux 6.3.2-arch1-1.

 One feat of symbolic execution of probabilistic programs that \toolname~implements, is that paths with an unsatisfiable \emph{observe} statement can be filtered out, since they have zero likelihood in the probabilistic model.
\toolname~therefore \emph{discards} such paths.

 For each experiment, we report the
 following metrics: the number of traces explored, the number of traces which were discarded
 due to a failed observe statement (`{--}' if there were no observe statements in the program), the maximum number
 of random samples drawn, the number of lines of code, and the
 time \toolname took to explore all paths.

 We make a few general observations about the results. Outside
 of the DieCond, Grass, NoisyOr, and Random Z2 Walk (4,8) examples,
 \toolname terminates in under a second. While these are relatively
 small examples, they still showcase a range of path counts and number
 of random samples drawn. Additionally, of the case studies that had observe statements present, many paths were
 discarded due to reaching an unsatisfiable observe statement.
 The Grass case study, for example, which involves six Bernoulli samples about weather conditions and such, had 36 paths discarded.
 In Section~\ref{sec:conclusions}, we
 discuss how we might utilize this information in future work to optimize the performance
 of probabilistic programs with observe statements.

 \section{Related Work} \label{sec:relwork}
 \vspace*{-1mm}
 Symbolic transition systems as described in this work have recently
 been formalized in a non-probabilistic setting by de Boer and
 Bonsangue \cite{deboer2021symbolic}.  From that starting point, this
 work extends to sampling in probabilistic programming by
 incorporating symbolic random variables $\{\y_0,\y_1,\dots\}$ in
 the symbolic substitutions and path conditions.  Furthermore, we
 introduced \emph{path observations} to keep track of observe
 statements.

 Denotational semantics for Bayesian inference is an active research
 area
 \cite{staton16semantics,staton17commutative,dahlqvist18borelkernelapprimations,holtzen2020scaling}. Mostly,
 the focus has been on \emph{discrete} probabilistic programs
 \cite{holtzen2020scaling,olmedo2018conditioning}, whereas we support
 sampling of both discrete \emph{and} continuous distributions,
 thereby generalizing the probabilistic choice based on discrete
 sampling used in these works.

Staton \cite{staton17commutative,staton16semantics} describes
\emph{observe-from} statements in his denotational semantics as an
encoding for a \emph{score} construct used to attach a likelihood
scoring to an execution trace.  The construct
$\syntax{observe}~\exp~\syntax{from}~\mathcal D$, where
$\mathcal D$ is some distribution, is then sugar for
$\syntax{score}~f_{\mathcal D}(\exp)$, where $f_{\mathcal D}$ is the
probability density function for the distribution $\mathcal D$.  The
score value in that semantics is therefore akin to the probability
measure of our path observation.  Staton's work considers language
semantics of a functional nature, where the difficulty mainly lies in
handling higher-order functions.  In an imperative language,
this is generally not a major concern.  This allows us to consider the
more general construct of observing Boolean formulas (of positive
measure), rather than observing fixed samples (which may have zero
measure).  To the best of our knowledge, our work provides the first
semantics for Boolean observe statements as a forward state
transformer (as proposed by Kozen \cite{kozen1979semantics}) for an
imperative probabilistic programming language in the presense of unbounded loops.

Sampson et al.~\cite{sampson2014expressing} used symbolic execution to
transform probabilistic programs into a Bayesian network.  Their
semantics was aimed at verification of certain probabilistic
assertions, however, and lacked a formal foundation in measure
theory. Luckow et al.\ combined symbolic
  execution with model counting to analyze programs with discrete
  distributions and nondeterministic choice \cite{luckow2014exact}. Their work used
  schedulers to reduce Markov Decision Processes to Markov Chains, in
  contrast to our work with continuous distributions and observe
  statements, founded in measure theory.  Susag et
al.~\cite{susag2022} considered symbolic execution for programs with
random sampling to automatically verify quantitative correctness
properties over unknown inputs.  While they also used symbolic random
variables, they were more focused on verifying correctness properties
of ``real-world'' programs (e.g., written in C++) that use random
sampling. They solely considered discrete distributions
and did not support observe statements.

Gehr et al.~\cite{gehr2016psi} developed PSI, \emph{probabilistic
  symbolic inference}, a tool that enables programmers to perform
posterior distribution, expectation, and assertion queries through
symbolic inference.  Their symbolic reasoning engine works on symbolic
representations of probability distributions, whereas we have symbolic
terms which are interpreted as \emph{random variables} of which we do
not know the distribution in principle.  By default, PSI computes a
symbolic representation of the joint posterior distribution
represented by the given probabilistic program.  In contrast,
\toolname explores all paths through a given probabilistic program and
reports on the path condition, substitution map, and path observation
of each path.  We see these two tools as complementary, as inference
is only one aspect of probabilistic programming.

\section{Conclusion and Future Work} \label{sec:conclusions}
We have defined new symbolic semantics for imperative probabilistic
programs supporting forward execution and conditioning (Bayesian
inference), which we proved to be a big-step semantics for symbolic execution of probabilistic programs.
To support Bayesian inference, we extended Kozen's denotational semantics to
include \lstinline{observe} statements of positive-measure events.
Significantly, the symbolic semantics thus theoretically supports implementation of a symbolic
executor for imperative probabilistic programs with continuous domain
variables.  We introduced \emph{path observations} to keep track of
the part of the initial state space that lead to acceptance of observe
statements for each symbolic trace.
The symbolic transition system is implemented in our
prototype tool \toolname. Its effectiveness has been demonstrated on
 example programs, producing results with path
conditions and path observations that correctly represent (prior) path
probabilities and likelihoods of the symbolic traces, as expected in
light of \Cref{thm:mainresult}.
Since path conditions represent priors and path observations represent likelihoods, we believe it to be beneficial to conceptually separate the two in a symbolic executor.

Interestingly, our semantics and its accompanying correctness results
enable one to prove the correctness of certain program
transformations.  More generally, we believe the symbolic semantics
has potential applications in model checking tools, deductive proof
systems, optimizations, and reasoning about termination.  For example,
one can exploit this theory to commute sample and observe statements
under appropriate substitutions.  In complex probabilistic models, the
placement of observe statements may then become an optimization
problem. We plan to explore such program transformations both
theoretically and experimentally in the future.
Moreover, we intend to investigate the limiting behavior of our correctness result, which contains an infinite sum in the presence of unbounded loops.
One naturally wonders how fast the sum approximates the true posterior of the program, and how different structures in a program influence the speed of this approximation.

A shortcoming of our semantics for \emph{observe} statements is that
zero-measure events cannot be observed in our language.
Denotationally, the measure would then always yield zero;
operationally, \emph{almost all} simulations will be aborted.
Zero-measure observed Boolean conditions may be included in future
work, for example by considering \emph{measure couplings} and
\emph{disintegrations} \cite{dahlqvist18borelkernelapprimations}.
\paragraph{Data availability.} The source code for \toolname and the experiments we performed with it are available on Zenodo~\cite{symProbArtifact2023}.
\bibliographystyle{splncs04}
\bibliography{ref.bib}

\appendix

\section{Proofs}
\subsection{Proof of \Cref{lem:substitutionlemmaexp}}
  \substitutionlemmaexp*
  \noindent For \emph{Boolean} expressions, we moreover have the following:
  \begin{lemma}[Substitution lemma for Boolean expressions]\label{lem:substitutionlemmabexp}
    For all $\bexp \in \B$, $\sinterpret{\app\initsub\bexp} = \bexpset$.
  \end{lemma}
  \begin{proof}
    By induction on the structure of expressions $\exp \in \E$ and Boolean expressions $\bexp \in \B$ as presented in \Cref{sec:syntax}.
    There are two base cases for $\exp$:
    \begin{itemize}
      \item For $\exp = \x_i$, we have $ \sinterpret {\x_i} (\rho) = \rho_i = \interpret {\x_i}(\res \rho n) $.
      \item For $\exp = q$: $ \sinterpret {q} (\rho) = q = \interpret {q} (\res \rho n) $.
    \end{itemize}
    There is an induction step for every operator $\op$ of arity $\arop$.
    Each induction step applies $\arop$ induction hypotheses:
    $$ \begin{array}{rll}
      \sinterpret{\op (\exp_1, \ldots, \exp_{\arop})}(\rho) = &
      \opsem ( \sinterpret{\exp_1}(\rho), \ldots, \sinterpret{\exp_{\arop}} (\rho)) & \textup{definition of } \sinterpret {\cdot} \\
      = & \opsem ( \interpret {\exp_1}(\res \rho n), \ldots, \interpret {\exp_{\arop}} (\res \rho n) ) & \textup {IH } \arop \textup { times} \\
      = & \interpret { \op ( \exp_1, \ldots, \exp_{\arop} ) } (\res \rho n) & \textup{definition of } \interpret {\cdot}
    \end{array}$$
    We proceed now with structural induction on $\bexp$.
    The base cases:
    \begin{itemize}
      \item $\bexp = \true$: $\rho \in \sinterpret {\initsub (\true)}=\sinterpret{\strue}$ if and only if $\rho \in (\interpret {\true} \times \R^\omega)$.
      \item $\bexp = \false$: $\sinterpret{\initsub(\false)} = \sinterpret{\sfalse}= \emptyset = \interpret{\false} \times \R^\omega$.
      \item $\bexp = \exp_1 \relop \exp_2$: $\rho \in \sinterpret {\initsub (\exp_1 \relop \exp_2)}$ if and only if $\sinterpret{\exp_1}(\rho) \relop \sinterpret{\exp_2}(\rho)$ (by def.) if and only if $\interpret {\exp_1} (\res \rho n) \relop \interpret {\exp_2} (\res \rho n)$ (by \Cref{lem:substitutionlemmaexp}) if and only if $\rho \in \interpret {\exp_1 \relop \exp_2} \times \R^\omega$ (by def.).

        We thus conclude this case with the identity $\sinterpret{\initsub (\exp_1 \relop \exp_2)} = \interpret {\exp_1 \relop \exp_2} \times \R^\omega$.
      \item $\bexp = \bexp_1 \por \bexp_2$: simply apply IH twice to show that their respective binary unions are equal, and then conclude that
        $$ \sinterpret {\initsub(\bexp_1 \por \bexp_2)} = \interpret {\bexp_1 \por \bexp_2} \times \R^\omega$$
      \item $\bexp = \bexp_1 \pand \bexp_2$ is analogous.
      \item $\bexp = \pneg \bexp_1$ idem.
    \end{itemize}
    This completes the proof.
  \end{proof}

\subsection{Proof of \Cref{thm:symbexsemantics}} \label{app:symbexsemanticsproof}
  Recall that $\Gamma_\p = \{(\sub,\idx,\pc,\po) \setst (\p,\initsub,0,\strue,\strue) \trCl (\pskip,\sub,\idx,\pc,\po) \}$.

  Throughout all the proofs in this appendix, the fixed initial symbolic configuration $(\initsub,0,\strue,\strue)$ is denoted by $\initconfig$ and $\cfg,\cfg',\cfg'',\cfg_1,\cfg_2$ respectively abbreviate configurations $(\sub,\idx,\pc,\po)$, $(\sub',\idx',\pc',\po')$, $(\sub'',\idx'',\pc'',\po'')$, $(\sub_1,\idx_1,\pc_1,\po_1)$, and $(\sub_2,\idx_2,\pc_2,\po_2)$.
  \symbexsemantics*
  \begin{proof}
    For every $\p$, define the mapping
    $$\Phi_\p : \Gamma_\p \to \FF_\p, \quad  (\sub,\idx,\pc,\po) \mapsto (\sinterpretk\sub\idx,\sinterpret\pc,\sinterpret\po) $$
    \Cref{prop:welldefined} says that this mapping is well-defined for all $\p$, i.e., $(\sinterpretk\sub\idx,\sinterpret\pc,\sinterpret\po)$ is actually an element of $\FF_\p$.
    Indeed, take $\q = \pskip$ and use that $\FF_{\pskip}$ is the singleton $\{ (\id_{\statespace},\statespace,\statespace)\}$.
    \Cref{prop:mappingsurjective} says that these mappings are all surjective.

    For injectivity, there are some technical subtleties to be addressed.
    Consider the following example program $\p$:
    $$ \syntax{if(x<0)\{ if(1=0) x:=42 else x:=0 \} else \{ if(1=0) x:=42 else x:=1 \}}$$
    The program contains an \emph{if} statement with a nested \emph{if} statement in both branches.
    Symbolic execution of this program will thus yield \emph{four} execution paths.
    \emph{Two} of these paths have an unsatisfiable path condition due to the Boolean tests $\syntax{1=0}$.
    Furthermore, these two paths \emph{both} transform the input variable $\x$ to the value $42$.
    Hence, the big-step semantics, in contrast, has \emph{three} elements:
    $$ \FF_\p = \{ (x \mapsto 42, \emptyset, \R^{1+\omega}), (x\mapsto 0, \{x<0\}, \R^{1+\omega}), (x\mapsto 1, \{x\geq 0\}, \R^{1+\omega}) \} $$
    where the two branches with unsatisfiable path conditions have ``collapsed'' to the first triple.
    Hence, technically speaking there is no bijection between $\Gamma_\p$ and $\FF_\p$ for this edge case.
    
    This problem could perhaps be solved by considering \emph{disjoint} union of branching in the symbolic semantics, or \emph{multisets}.
    But this would also require loop iterations and sequencing to be done in a disjoint union manner and would result in a lot of bookkeeping.
    To avoid this, we choose instead to identify all program executions (on the side of $\Gamma_\p$ as well as on the side of $\FF_\p$) that have unsatisfiable path conditions.
    This does not influence the well-definedness property, nor the established surjectivity.

    By analysis of the rules in \Cref{fig:symos}, two traces $(\p,\initconfig)\trCl(\pskip,\cfg)$ and $(\p,\initconfig)\trCl(\pskip,\cfg')$ in $\Gamma_\p$ are different iff there is a symbolic Boolean expression $\sbexp$ such that $\pc$ contains $\sbexp$ and $\pc'$ contains $\sneg \sbexp$ as a conjunct.
    This immediately entails $\sinterpret\pc \cap \sinterpret{\pc'} = \emptyset$.
    Then, either both are empty, hence $\pc$ and $\pc'$ are unsatisfiable and therefore considered equal, or their image under $\Phi_\p$ is distinct.
    This shows injectivity of all mappings $\Phi_\p$.
  \end{proof}

  \begin{proposition}[Multi-step backward assimilation into symbolic semantics] \label{prop:welldefined}
    For all $\p,\q\in\P$, if $(\p,\initconfig)\trCl(\q,\sub,\idx,\pc,\po)$ then for all $(F,B,\OO) \in \FF_{\q}$:
    \begin{equation} \label{eq:symbolicmappinggoal}
    (F \circ \sinterpretk\sub\idx,\sinterpret {\pc} \cap {\sinterpretk\sub\idx}^{-1}[B] ,\sinterpret{\po} \cap {\sinterpretk\sub\idx}^{-1}[\OO]) \in \FF_{\p} \end{equation}
  \end{proposition}
  \begin{proof}
    The proof is by induction on the length of the transition chain $(p,\initconfig) \trCl (S',\sub,\idx,\pc,\po)$, where, in the inductive step, we analyze the \emph{first} rule that was used.

    If $(\p,\initconfig) \trCl (\q,\cfg)$ has length zero (from reflexivity) then we know that $\q = \p$ and $\cfg = \initconfig$.
        With this, we may observe that ${\sinterpretk\sub\idx} = \sinterpretk{\initsub}0 = \id_{\statespace}$, and $\sinterpret \pc = \sinterpret \po = \sinterpret \strue = \statespace$.
        Now let $(F,B,\OO) \in \FF_{\q}$.
        The result follows since the triple in \eqref{eq:symbolicmappinggoal} in this case is just $(F,B,\OO)$, and $(F,B,\OO) \in \FF_{\q} = \FF_{\p}$.
    For the inductive step let $(F,B,\OO) \in \FF_{\q}$ be arbitrary and let
    $$(\p,\initconfig) \tr (\p',\cfg') \trCl (\q,\cfg) $$
    be a transition chain of length $\l+1$.
    We cannot immediately apply IH to the transition chain of length $\l$, since $\cfg'$ may not be the initial configuration.
    Instead, we use \Cref{lem:canonicalexecution} to obtain the chain
    $$ (\p',\cfg_0) \trCl (\q,\cfg'') $$
    of length $\l$ such that $\cfg''$ has the following properties:
    \begin{enumerate}[label=(\roman*)]
      \item ${\sinterpretk\sub\idx} = {\sinterpretk{\sub''}{\idx''}} \circ {\sinterpretk{\sub'}{\idx'}} $;
      \item $\sinterpret \pc = \sinterpret {\pc'} \cap {\sinterpretk{\sub'}{\idx'}}^{-1}[\sinterpret{\pc''}]$; and
      \item $\sinterpret \po = \sinterpret {\po'} \cap {\sinterpretk{\sub'}{\idx'}}^{-1}[\sinterpret{\po''}]$.
    \end{enumerate}
    Now, by IH, since $(\p',\initconfig) \trCl (\q,\cfg'')$ is a chain of length $\l$ and $(F,B,\OO) \in \FF_{\q}$:
    $$(F',B',\OO') := (F \circ \sinterpretk{\sub''}{\idx''},\sinterpret {\pc''} \cap {\sinterpretk{\sub''}{\idx''}}^{-1}[B] ,\sinterpret{\po''} \cap {\sinterpretk{\sub''}{\idx''}}^{-1}[\OO]) \in \FF_{\p'}$$
    From the one-step \Cref{lem:symbolicmappingsmallstep}, using the transition $(\p,\initconfig) \tr (\p',\cfg')$, then
    $$ (F' \circ {\sinterpretk{\sub'}{\idx'}}, 
        {\sinterpretk{\sub'}{\idx'}}^{-1} [ B' ] \cap \sinterpret{\pc'}, 
        {\sinterpretk{\sub'}{\idx'}}^{-1} [ \OO' ] \cap \sinterpret{\po'})
        \in \FF_{\p} $$
    We rewrite each of the three components appropriately as follows:
    \begin{enumerate}[label=(\roman*)]
      \item $F' \circ {\sinterpretk{\sub'}{\idx'}} = F \circ {\sinterpretk{\sub''}{\idx''}} \circ {\sinterpretk{\sub'}{\idx'}} = F \circ {\sinterpretk\sub\idx}$;
      \item ${\sinterpretk{\sub'}{\idx'}}^{-1}\Big[\sinterpret {\pc''} \cap {\sinterpretk{\sub''}{\idx''}}^{-1}[B]\Big] \cap \sinterpret{\pc'} = \underbrace{{\sinterpretk{\sub'}{\idx'}}^{-1}[\sinterpret{\pc''}] \cap \sinterpret{\pc'}}_{\sinterpret{\pc}} \cap \underbrace{{\sinterpretk{\sub'}{\idx'}}^{-1}[{\sinterpretk{\sub''}{\idx''}}^{-1}[B]]}_{{\sinterpretk\sub\idx}^{-1}[B]}$, and this is just $\sinterpret {\pc} \cap \sinterpretk{\sub}{\idx}^{-1}[B]$.
      \item Analogous to the previous item, ${\sinterpretk{\sub'}{\idx'}}^{-1} [ \OO' ] \cap \sinterpret{\po'} = \sinterpret \po \cap \sinterpretk{\sub}{\idx}^{-1}[\OO]$.
    \end{enumerate}
    We have thus verified that
    $$ (F \circ \sinterpretk\sub\idx, \sinterpret {\pc} \cap \sinterpretk{\sub}{\idx}^{-1}[B], \sinterpret \po \cap \sinterpretk{\sub}{\idx}^{-1}[\OO]) \in \FF_{\p} $$
    and so the inductive step, and the whole proof with it, is finished.
  \end{proof}
  \begin{lemma}[One-step backward assimilation into symbolic semantics] \label{lem:symbolicmappingsmallstep}
    If $(\p,\initconfig) \tr (\q,\cfg)$ and $(F,B,\OO) \in \FF_{\q}$ then
    \begin{equation}\label{eq:smallstepgoal}
    (F \circ \sinterpretk{\sub}{\idx}, \sinterpretk{\sub}{\idx}^{-1}[B] \cap \sinterpret {\pc}, \sinterpretk{\sub}{\idx}^{-1}[\OO] \cap \sinterpret {\po}) \in \FF_{\p}
    \end{equation}
  \end{lemma}
  \begin{proof}
    By induction on the height of the proof tree that justifies the transition rule $(p,\initconfig) \tr (\q,\cfg)$.
    All except \refRule{seq-n} are base cases.
    \begin{itemize}
      \item For \refRule{asgn}, we have the following transition from the initial configuration:
        $$(\assign,\initsub,0,\strue,\strue) \tr (\pskip,\initsub\update{\x_i}{\app\initsub\exp},0,\strue,\strue)$$
        By definition of $\FF_{\pskip}$, which is the singleton $\{ (\id_{\statespace},\statespace,\statespace)$, it must be that $F = \id_{\statespace}$, and $B = \OO = \statespace$.
        Thus, since $\app\initsub\exp=\exp$, with $\sub = \initsub \update {i}\exp$ and $\idx=0$, we need to verify that
        $$ (\id_\statespace \circ \sinterpretk{\sub}{\idx}, {\sinterpretk{\sub}{\idx}}^{-1} [\statespace] \cap \sinterpret{\strue}, {\sinterpretk{\sub}{\idx}}^{-1}[\statespace] \cap \sinterpret{\strue}) \in \FF_{\assign} $$
        and recall that $ \FF_{\assign} = \{ (\rho \mapsto \rho \update{i}{\interpret{\exp}(\res \rho n)},\statespace,\statespace)\} $.
        Since $\sinterpretk \sub \idx$ is total and $\sinterpret {\strue} = \statespace$, the only thing to show is that $\sinterpretk{\sub}\idx : \rho \mapsto \rho \update i {\interpret{\exp}(\res \rho n)}$.
        Let us look at $\sinterpretk{\sub}\idx (\rho)$:
        $$ \begin{array}{lr}
          ( \sinterpret{\sub(\x_0)}(\rho), \ldots, \sinterpret{\sub(\x_{n-1})}(\rho), \rho_n, \rho_{n+1}, \ldots )  & \textup{def. of } \sinterpretk\sub\idx \\
          \quad =
          ( \sinterpret{\x_0}(\rho), \ldots, \sinterpret{\x_{i-1}}(\rho), \sinterpret{\exp}(\rho), \\
          \quad \qquad
          \sinterpret{\x_{i+1}}(\rho), \ldots, \sinterpret{\x_{n-1}}(\rho), \rho_n, \rho_{n+1}, \ldots ) & \textup {def. of } \sub \\
          \quad =
          ( \rho_0, \ldots, \rho_{i-1}, \sinterpret{\exp}(\rho), \rho_{i+1}, \ldots, \rho_{n-1}, \rho_n, \rho_{n+1}, \ldots ) & \quad \textup{def. of } \sinterpret{\cdot} \\
          \quad =
          ( \rho_0, \ldots, \rho_{i-1}, \interpret {\exp} (\res \rho n), \rho_{i+1}, \ldots, \rho_{n-1}, \rho_n, \rho_{n+1}, \ldots ) & \textup{\Cref{lem:substitutionlemmaexp}} \\
          \quad = \rho \update i {\interpret {\exp} (\res \rho n)} & \textup{notation}
        \end{array}$$
        This concludes the case.

      \item We continue with \refRule{smpl}, for which we have the following transition:
        $$(\psample,\initconfig) \tr (\pskip,\initsub \update {\x_i} {\y_{0}},1,\strue,\strue)$$
        and by definition, we have $ \FF_{\psample} = \{(\rho \mapsto \sample i (\rho),\statespace,\statespace) \}$ and $\FF_{\pskip} = \{ (\id_{\statespace}, \statespace ,\statespace) \}$.
        Similar to the previous case, here it suffices to show that $\sinterpretk \sub 1 (\rho) = \sample i (\rho)$ (where now $\sub = \initsub\update{i}{\y_0}$):
        $$ \begin{array}{rlr}
          \sinterpretk \sub 1 (\rho)
          = & ( \sinterpret{\sub(\x_0)}(\rho), \ldots, \sinterpret{\sub(\x_{n-1})}(\rho), \rho_{n+1}, \rho_{n+2}, \ldots )  & \textup{def. of } \sinterpretk\sub\idx \\
          = & ( \sinterpret{\x_0}(\rho), \ldots, \sinterpret{\x_{i-1}}(\rho), \sinterpret{\y_0}(\rho), \\
            & \sinterpret{\x_{i+1}}(\rho), \ldots, \sinterpret{\x_{n-1}}(\rho), \rho_{n+1}, \rho_{n+2}, \ldots ) & \textup {def. of } \sub \\
          = & ( \rho_0, \ldots, \rho_{i-1}, \rho_n, \rho_{i+1}, \ldots, \rho_{n-1}, \rho_{n+1}, \rho_{n+2}, \ldots ) & \quad \textup{def. of } \sinterpret{\cdot} \\
          = & \sample i (\rho) & \textup{notation}
        \end{array}$$
        This case is now finished.

      \item \refRule{obs}: the transition has the following shape:
        $$ (\observe,\initconfig) \tr (\pskip,\initsub,0,\strue,\strue \sand \initsub(\bexp)) $$
        Again, $(F,B,\OO) \in \FF_{\pskip}$ means $F = \id_{\statespace}$ and $B=\OO=\statespace$.
        Note that $\FF_{\observe} = \{ (\id_{\statespace},\statespace,\bexpset) \}$.
        \begin{enumerate}[label=(\roman*)]
          \item $F \circ \sinterpretk {\initsub}0 = \id_{\statespace} \circ \id_{\statespace} = \id_{\statespace}$;
          \item $\sinterpretk {\initsub}0^{-1}[B] \cap \sinterpret{\pc} = \sinterpretk {\initsub}0^{-1}[\statespace] \cap \sinterpret{\strue} = \statespace \cap \statespace = \statespace$; and
          \item $\sinterpretk {\initsub}0^{-1}[\OO] \cap \sinterpret{\po} = \sinterpretk {\initsub}0^{-1}[\statespace] \cap \sinterpret{\strue \sand \initsub(\bexp)} = \statespace \cap \sinterpret {\strue} \cap \sinterpret{\initsub(\bexp)} = \bexpset$.
        \end{enumerate}
        Here, we used \Cref{lem:substitutionlemmabexp} in the last equality.
        This case is now finished.

      \item For \refRule{seq-0}, we have the following transition:
        $$ (\seq\pskip\p,\initconfig) \tr (\p,\initsub,0,\strue,\strue)  $$
        We recognize that if $(F,B,\OO) \in \FF_{\p}$, then
        $$(F \circ \sinterpretk{\initsub}0, \sinterpretk{\initsub}0^{-1}[B] \cap \sinterpret {\strue}, \sinterpretk{\initsub}0^{-1}[\OO] \cap \sinterpret{\po}) = (F,B,\OO)$$
        Now, since $(F,B,\OO)\in\FF_{\p}$ and $(\id_{\statespace},\statespace,\statespace)\in\FF_{\pskip}$, by definition, $\FF_{\seq\pskip\p}$ also contains $(F,B,\OO)$, so this case is done.

      \item \refRule{seq-n}: the only inductive step.
        The derivation tree for the transition is the following:
        $$ \srule{}{(\p,\initconfig) \tr (\p',\cfg)}{(\seq\p\q,\initconfig) \tr (\seq{\p'}\q,\cfg)}$$
        The induction hypothesis is the following: if $(\p,\initconfig)\tr(\p',\cfg)$ and $(F,B,\OO)\in\FF_{\p'}$ then $(F\circ\sinterpretk\sub\idx,\sinterpretk\sub\idx^{-1}[B]\cap\sinterpret{\pc},\sinterpretk\sub\idx^{-1}[\OO] \cap \sinterpret{\po})$.

        Now let $(F,B,\OO)\in\FF_{\seq{\p'}\q}$.
        Then, by definition, there are $(F_1,B_1,\OO_1)\in\FF_{\p'}$ and $(F_2,B_2,\OO_2)\in\FF_{\q}$ such that \rom 1 $F = F_2 \circ F_1$, \rom 2 $B = B_1 \cap F_1^{-1}[B_2]$, and \rom 3 $\OO = \OO_1 \cap F_1^{-1}[\OO_2]$.
        We can apply the IH to $(F_1,B_1,\OO_1) \in \FF_{\p'}$ with the transition $(\p,\initconfig)\tr(\p',\cfg)$ to learn that
        $$(F_1 \circ \sinterpretk{\sub}{\idx}, \sinterpretk{\sub}{\idx}^{-1}[B_1] \cap \sinterpret {\pc}, \sinterpretk{\sub}{\idx}^{-1}[\OO_1] \cap \sinterpret {\po}) \in \FF_{\p} $$
        But then, by definition, $\FF_{\seq\p\q}$ contains $(F',B',\OO')$, where
        \begin{enumerate}[label=(\roman*)]
          \item $ F' = F_2 \circ (F_1 \circ \sinterpretk{\sub}{\idx}) = F \circ \sinterpretk{\sub}{\idx}$;
          \item $ B' = \sinterpretk{\sub}{\idx}^{-1}[B_1] \cap \sinterpret {\pc} \cap (F_1 \circ \sinterpretk{\sub}{\idx})^{-1} [B_2] = \sinterpretk{\sub}{\idx}^{-1}[B] \cap \sinterpret{\pc}$; and
          \item $ \OO' = \sinterpretk{\sub}{\idx}^{-1}[\OO_1] \cap \sinterpret {\po} \cap (F_1 \circ \sinterpretk{\sub}{\idx})^{-1} [\OO_2] = \sinterpretk{\sub}{\idx}^{-1}[\OO] \cap \sinterpret{\po}$.
        \end{enumerate}
        This is exactly what we need to show in this case!

      \item \refRule{if-T}: analogous to \refRule{if-F}.
      \item For \refRule{if-F}, we have the following shape of the transition:
        $$ (\pif\bexp{\p_1}{\p_2},\initconfig) \tr (\p_2,\initsub,0,\strue \sand \sub(\pneg b),\strue) $$
        Now let $(F,B,\OO) \in \FF_{\p_2}$.
        Using that $\sinterpretk{\initsub}{0}=\id_\statespace$, we recognize that
        \begin{enumerate}[label=(\roman*)]
          \item $F \circ \sinterpretk {\initsub}0 = F$;
          \item $\sinterpretk {\initsub}0^{-1}[B] \cap \sinterpret{\strue \sand \initsub(\pneg \bexp)} = B \cap (\compl{\interpret\bexp} \times\R^\omega)$, by \Cref{lem:substitutionlemmabexp}; and
          \item $\sinterpretk {\initsub}0^{-1}[\OO] \cap \sinterpret{\strue} = \OO$.
        \end{enumerate}
        Indeed, by definition, $(F,B \cap (\bexpcset),\OO) \in \FF_{\pif {\bexp}{\p_1}{\p_2}}$, so the case is finished.

      \item \refRule{iter-F}: we have the transition rule
        $$ (\while\bexp\p,\initconfig) \tr (\pskip,\initsub,0,\strue \sand \initsub(\pneg \bexp),\strue) $$
        and $(F,B,\OO) \in \FF_{\pskip}$ means $F = \id_{\statespace}$ and $B = \OO = \statespace$.
        With these data, we recognize that
        \begin{enumerate}[label=(\roman*)]
          \item $F \circ \sinterpretk {\initsub}0 = \id_{\statespace}$;
          \item $\sinterpretk {\initsub}0^{-1} [B] \cap \sinterpret {\strue \sand \initsub(\pneg \bexp)} = \bexpcset$ (\Cref{lem:substitutionlemmabexp}); and
          \item $\sinterpretk {\initsub}0^{-1} [\OO] \cap \sinterpret {\strue} = \statespace$.
        \end{enumerate}
        Recall the definition of $\FF_{\while\bexp\p}$:
        $$ \bigcup_{m=0}^\infty \mathbb F_{\bexp,\p}^m \{ (\id_{\statespace},\bexpcset,\statespace) $$
        Consider $m=0$ to finish the case.

      \item \refRule{iter-T}: we have the transition rule in the following shape:
        $$ (\while\bexp\p,\initconfig)\tr(\seq\p{\while\bexp\p},\initsub,0,\strue \sand \app\initsub\bexp,\strue) $$
        Consider $(F,B,\OO) \in \FF_{\seq\p{\while\bexp\p}}$.
        By definition, then, there is $(F_\p,B_\p,\OO_\p) \in \FF_{\p}$ and $(F',B',\OO') \in \FF_{\while\bexp\p}$ such that \rom{1} $F = F' \circ F_\p$, \rom{2} $B = B_\p \cap F_\p^{-1}[B']$, and \rom{3} $\OO = \OO_\p \cap F_\p^{-1}[\OO']$.
        Furthermore, by definition of $\FF_{\while\bexp\p}$, there is $m \in \N$ such that $(F',B',\OO') \in \mathbb F_{\bexp,\p}^m \{ (\id_{\statespace},\bexpcset,\statespace) \}$.
        We now recognize that
        \begin{enumerate}[label=(\roman*)]
          \item $F \circ \sinterpretk {\initsub}0 = F = F' \circ F_\p$;
          \item $\sinterpretk {\initsub}0^{-1}[B] \cap \sinterpret {\strue \sand \initsub(b)} = B_\p \cap F_\p^{-1}[B'] \cap (\bexpset)$ (\Cref{lem:substitutionlemmabexp}); and
          \item $\sinterpretk {\initsub}0^{-1}[\OO] \cap \sinterpret {\strue} = \OO = \OO_\p \cap F_\p^{-1}[\OO']$.
        \end{enumerate}
        Now, since $(F',B',\OO') \in \mathbb F_{\bexp,\p}^m \{ (\id_{\statespace},\bexpcset,\statespace)\}$ and $(F_\p,B_\p,\OO_\p) \in \FF_{\p}$, it follows by definition of $\mathbb F_{\bexp,\p}$ that
        $$\scalebox{0.9}{$ (F' \circ F_\p, (\bexpset) \cap B_\p \cap F_\p^{-1}[B'], \OO_\p \cap F_\p^{-1}[\OO']) 
        \in \mathbb F_{\bexp,\p}^{m+1} \{ (\id_{\statespace},\bexpcset, \statespace) \} $ }$$
        which is a subset of $\FF_{\while\bexp\p}$, so we are done.

    \end{itemize}
    Induction has finished for the small-step semantics $\tr$, and so the proof is done.
  \end{proof}
  \begin{lemma} [\emph{Canonical Symbolic Execution}] \label{lem:canonicalexecution}
    For all statements $\p$ and configurations $\cfg$, $(\p,\cfg)\trCl(\q,\cfg')$ if and only if $(\p,\initconfig) \trCl (\q,\cfg_1)$ (of the same length) such that:
    \begin{enumerate}[label=(\roman*)]
      \item\label{item:propsub} $\sinterpretk{\sub'}{k'} = \sinterpretk{\sub_1}{k_1} \circ \sinterpretk{\sub}{k}$,
      \item\label{item:propk}   $k' = k_1 + k$,
      \item\label{item:proppc}  $\sinterpret{\pc'} = \sinterpret{\pc} \cap \sinterpretk{\sub}{k}^{-1}[\sinterpret{\pc_1}]$, and
      \item\label{item:proppo}  $\sinterpret{\po'} = \sinterpret{\po} \cap \sinterpretk{\sub}{k}^{-1}[\sinterpret{\po_1}]$.
    \end{enumerate}
    In case $\cfg$, $\cfg'$, and $\cfg_1$ satisfy properties \Cref{item:propsub,item:propk,item:proppc,item:proppo}, we write $\cfg' \cfgeq \cfg_1 \cfgcmp \cfg$.
    We call $(\p,\initconfig)\trCl(\q,\cfg_1)$ a \emph{canonical} symbolic execution from $\p$ to $\q$.
  \end{lemma}
  \begin{proof}
    First some observations.
    If $\cfg' \cfgeq \cfg_1 \cfgcmp \cfg$ then $\sinterpretk{\sub'}{k'} (\rho) = \sinterpretk{\sub_1}{k_1} \big( \sinterpretk{\sub}{\idx} (\rho)\big)$ for all $\rho$, and we have the following four facts that we will use throughout the proof:
    \begin{enumerate}[label=(\Alph*)]
      \item\label{item:substcompvariable} For all $\rho \in \R^{n+\omega}$ and $i \in \{0,\ldots,n-1\}$, $\sinterpret {\sub'(\x_i)} (\rho) = \sinterpret{\sub_1(\x_i)}( \sinterpretk\sub\idx (\rho))$;
      \item\label{item:substcompsample} For all $\rho \in \R^{n+\omega}$ and $\l \in \N$, 
        $$\sinterpretk{\sub'}{k'}(\rho) (n+\l) = \rho_{k_1+k+n+\l} = \sinterpretk{\sub_1}{k_1}(\rho)(k+n+\l) = \sinterpretk{\sub_1}{k_1} (\sinterpretk{\sub}{\idx} (\rho)) (n+\l) $$
      \item\label{item:substcompexp}
        For all expressions $e \in \E$ and all $\rho \in \statespace$:
        $$ \sinterpret {\sub_1(\exp)}(\rho) = \sinterpret {\sub'(\exp)} (\sinterpretk{\sub}{\idx}(\rho)) $$
      \item\label{item:substcompbexp}
        For all Boolean expressions $\bexp \in \B$ we have the following equality of sets:
        $$ \sinterpretk{\sub}{\idx}^{-1}[\sinterpret{\sub_1(\bexp)}] = \sinterpret{\sub'(\bexp)} $$
    \end{enumerate}
    \Cref{item:substcompvariable,item:substcompsample} are element-wise equalities of $\sinterpretk{\sub'}{\idx'}(\rho) = \sinterpretk{\sub_1}{\idx_1}(\sinterpretk\sub\idx(\rho))$.
    \Cref{item:substcompexp} is proven by induction on the structure of expressions:
    \begin{itemize}
      \item The base case $\x_i$ is \Cref{item:substcompvariable}.
      \item The case $\exp = q \in \Q$ is trivial.
      \item The inductive step with $\exp = \op (\exp_1, \ldots, \exp_{\arop})$ is as follows:
        $$ \begin{array}{rl}
          \sinterpret {\sub_1(\op (\exp_1, \ldots, \exp_{\arop}))}(\sinterpretk\sub\idx(\rho))
          = & \opsem ( \sinterpret{\sub_1(\exp_1)}(\sinterpretk\sub\idx(\rho)), \ldots, \sinterpret{\sub_1(\exp_{\arop})}(\sinterpretk\sub\idx(\rho))) \\
          \overset{\textup{IHs}}= & \opsem (\sinterpret{\sub'(\exp_1)}(\rho), \ldots, \sinterpret{\sub'(\exp_{\arop})}(\rho)) \\
          = & \sinterpret {\sub'(\op (\exp_1, \ldots, \exp_{\arop}))} (\rho)
        \end{array}$$
    \end{itemize}
    \Cref{item:substcompbexp} is proven by induction on the structure of Boolean expressions:
    \begin{itemize}
      \item $\bexp=\true$: both sides equal $\statespace$ (since $\sinterpretk\sub\idx$ is total).
      \item $\bexp=\false$: both sides equal $\emptyset$.
      \item $\bexp=\exp_1 \relop \exp_2$: write $\rho' := \sinterpretk\sub\idx (\rho)$. Then
        $$ \begin{array}{rll}
          \rho \in \sinterpretk\sub\idx^{-1}[\sinterpret{\sub_1(\exp_1 \relop \exp_2)}]
          \iff & \rho' \in \sinterpret{\sub_1(\exp_1) \relop \sub_1(\exp_2)} &\\
          \iff & \sinterpret{\sub_1(\exp)}(\rho') \relop \sinterpret {\sub_1(\exp_2)} (\rho') \qquad \qquad  &\\
          \iff & \sinterpret{\sub'(\exp_1)}(\rho) \relop \sinterpret {\sub'(\exp_2)} (\rho) & \textup{\Cref{item:substcompexp}} \\
          \iff & \rho \in \sinterpret{\sub'(\exp_1 \relop \exp_2)} &
        \end{array}$$
      \item $\bexp=\bexp_1 \por \bexp_2$: then
        $$ \rho \in \sinterpretk\sub\idx^{-1}[\sinterpret{\sub_1(\bexp_1 \por \bexp_2)}] \iff \sinterpretk\sub\idx(\rho) \in \sinterpret{\sub_1(\bexp_1)} \cup \sinterpret{\sub_1(\bexp_2)} $$
        and the two IHs will conclude this case.
      \item $\bexp=\bexp_1 \pand \bexp_2$, and
      \item $\bexp=\pneg \bexp_1$ are analogous to the case for $\por$.
    \end{itemize}
  The proof of the lemma is by induction on the length of the transition chain, with an analysis of the \emph{last} rule that was used.
  The base case (length of zero) is from the reflexive closure, which tells us that $\cfg_1 = \cfg_0$
  for the canonical execution, and $\gamma' = \gamma$ for all other $\gamma$.
  With these data, \Crefrange{item:propsub}{item:proppo} are thus verified:
  \begin{enumerate}[label=(\roman*)]
    \item $\sinterpretk{\sub_1} {k_1} \circ \sinterpretk{\sub}{k} = \sinterpretk{\initsub}{0} \circ \sinterpretk{\sub}{\idx} =
     \id_{\statespace} \circ \sinterpretk\sub\idx = \sinterpretk\sub\idx = \sinterpretk{\sub'}{k'}$.
    \item $k_1 + k = 0 + k = k = k'$,
    \item $\sinterpret{\pc} \cap \sinterpretk{\sub}{k}^{-1}[\sinterpret{\pc_1}] = \sinterpret{\pc'} \cap \sinterpretk\sub\idx^{-1}[\sinterpret{\strue}] = \sinterpret{\pc'} \cap \sinterpretk\sub\idx^{-1}[\statespace] = \sinterpret{\pc'} \cap \statespace = \sinterpret{\pc'}$.
    \item Similar to the above item.
  \end{enumerate}
  This finishes the base case for the induction.
  For the inductive step, we prove the statement for the pairs of transition chains
  $$ (\p,\initconfig)\trCl(\q,\cfg_1)\tr(\r,\cfg_2) \qquad (\p,\cfg)\trCl(\q,\cfg')\tr(\r,\cfg'') $$
  of length $\l+1$, where IH gives us that $\cfg' \cfgeq \cfg_1 \cfgcmp \cfg$,
  and the goal is to show that $\cfg'' \cfgeq \cfg_2 \cfgcmp \cfg$.
  This in turn is done by induction on the height of the proof tree that justifies the two transitions
  $$ (\q,\cfg_1)\tr(\r,\cfg_2)\qquad(\q,\cfg')\tr(\r,\cfg'')$$
  So we continue with a case analysis of the rules in \Cref{fig:symos} that may justify the outgoing transition from $\q$---all except \refRule{seq-n} are base cases:
    \begin{itemize}
      \item For \refRule{asgn}, we have the following data:
        \begin{enumerate}[label=(\roman*)]
          \item $\sub_2 = \sub_1 \update {i}{\sub_1(\exp)}$ and $\sub'' = \sub' \update {i}{\sub'(\exp)}$;
          \item $k_2 = k_1$ and $k'' = k'$ (so by IH: $k_2 + k = k_1 + k = k' = k''$);
          \item $\pc_2 = \pc_1$ and $\pc'' = \pc'$.
            Then, by IH:
            $$\sinterpret{\pc} \cap \sinterpretk{\sub}{\idx}^{-1}[\sinterpret{\pc_2}] = \sinterpret{\pc} \cap \sinterpretk{\sub}{\idx}^{-1}[\sinterpret{\pc_1}] = \sinterpret{\pc'} = \sinterpret{\pc''}$$
          \item Analogously to the above item we see that $\sinterpret{\po} \cap \sinterpretk\sub\idx^{-1}[\sinterpret{\po_2}] = \sinterpret{\po''}$.
        \end{enumerate}
        The only non-trivial item we have to verify to conclude $\cfg'' \cfgeq \cfg_2 \cfgcmp \cfg$ is that
        \begin{equation} \label{eq:canonicalassigngoal}
          \sinterpretk{\sub''}{k''} = \sinterpretk{\sub_2}{k_2}\circ \sinterpretk{\sub}{k}
        \end{equation}
        Let $\rho \in \statespace$ be arbitrary.
        Then
        \begin{equation} \label{eq:canonicalassignLHS}
        \begin{array}{rl}
          \sinterpretk{\sub''}{k''}(\rho)
          & = (\sinterpret{\sub''(\x_0)}(\rho), \ldots, \sinterpret{\sub''(\x_{n-1})}(\rho), \rho_{n+k''}, \rho_{n+k''+1}, \dots ) \\
          & = ( \sinterpret{\sub'(\x_0)}(\rho), \ldots, \sinterpret{\sub'(\exp)}(\rho), \ldots, \qquad \qquad \footnotesize{(i\textup{-th var.})}\\
          & \quad\qquad \sinterpret{\sub'(\x_{n-1})}(\rho), \rho_{n+k'}, \rho_{n+k'+1}, \ldots ) \\
          & \overset{\textup{IH}}= ( \sinterpret{\sub_1(\x_0)}(\rho'), \ldots, \sinterpret{\sub'(\exp)}(\rho), \ldots, \qquad \qquad \footnotesize{(i\textup{-th var.})}\\
          & \quad\qquad \sinterpret{\sub_1(\x_{n-1})}(\rho'), \rho'_{n+k_1}, \rho'_{n+k_1+1}, \ldots )
        \end{array}
        \end{equation}
        Here, we have put $\rho' := \sinterpretk\sub\idx (\rho)$ and used \Cref{item:substcompvariable,item:substcompsample} for all elements except the $i$-th.
        On the other hand, we have, for the RHS of the goal \eqref{eq:canonicalassigngoal}:
        $$ \begin{array}{rl}
          \sinterpretk{\sub_2}{k_2}(\rho')
          & = ( \sinterpret{\sub_2(\x_0)}(\rho'), \ldots, \sinterpret{\sub_2(\x_{n-1})}(\rho'), \rho'_{n+k_2}, \rho'_{n+k_2+1}, \ldots ) \\
          & = ( \sinterpret{\sub_1(\x_0)}(\rho'), \ldots, \sinterpret{\sub_1(\exp)}(\rho'), \ldots \quad\qquad (i\textup{-th var.})\\
          & \quad\qquad \sinterpret{\sub_1(\x_{n-1})}(\rho'), \rho'_{n+k_1}, \rho'_{n+k_1+1}, \ldots )
        \end{array}$$
        So we may identify this with \eqref{eq:canonicalassignLHS} if indeed $\sinterpret{\sub'(\exp)}(\rho) = \sinterpret{\sub_1(\exp)}(\rho')$, and this follows from \Cref{item:substcompexp}, using IH, i.e., $\cfg'\cfgeq\cfg_1\cfgcmp\cfg$.
      \item We continue with \refRule{smpl}, for which we have the following data:
        \begin{enumerate}[label=(\roman*)]
          \item $\sub_2 = \sub_1 \update {i}{\y_{k_1}}$ and $\sub'' = \sub' \update {i}{\y_{k'}}$;
          \item $k_2 = k_1+1$ and $k'' = k'+1$, from which we deduce
            $$k_2 + k = k_1 + 1 + k = (k_1+k)+1 \overset{\textup{IH}}= k'+1 = k''$$
          \item and \rom 4: we have $\pc'' = \pc'$ and $\pc_2 = \pc_1$, and similarly for the path observations.
            Apply the same reasoning as in the case \refRule{asgn} to observe that
            $$\sinterpret {\pc''} = \sinterpret {\pc} \cap \sinterpretk\sub\idx^{-1}[\pc_2]
            \qquad \textup{and} \qquad
            \sinterpret {\po''} = \sinterpret {\po} \cap \sinterpretk\sub\idx^{-1}[\po_2]$$
        \end{enumerate}
        Again, the only truly interesting item we have to verify is that
          $\sinterpretk{\sub''}{k''} = \sinterpretk{\sub_2}{k_2}\circ \sinterpretk{\sub}{k}$.
        This is handled in exactly the same way as in the case for \refRule{asgn} with the following exception:
        instead of having to show $\sinterpret{\sub_1(\exp)}(\rho') = \sinterpret{\sub'(\exp)}(\rho)$, we have to prove that $\sinterpret{\y_{k_1}}(\rho') = \sinterpret{\y_{k'}}(\rho)$ for all $\rho \in \statespace$.
        But this is easy:
        $$\sinterpret{\y_{k_1}}(\rho') = \rho'_{n+k_1} = \rho_{n+k_1+k} = \rho_{n+k'} = \sinterpret{\y_{k'}}(\rho)$$
        Thus, we concluded this case since $\sinterpretk{\sub''}{k''} = \sinterpretk{\sub_2}{k_2} \circ \sinterpretk{\sub}{k}$ and so $\cfg'' \cfgeq \cfg_2 \cfgcmp \cfg$.

      \item We continue with \refRule{obs}, for which we have the following data:
        \begin{enumerate}[label=(\roman*)]
          \item $\sub_2 = \sub_1$ and $\sub'' = \sub'$, and thus (using \rom 2):
            $$ \sinterpretk{\sub_2}{k_2} \circ \sinterpretk{\sub}{\idx} = \sinterpretk{\sub_1}{k_1} \circ \sinterpretk{\sub}{\idx} = \sinterpretk{\sub_1}{k_1} \circ \sinterpretk{\sub}{\idx} = \sinterpretk{\sub'}{k'} = \sinterpretk{\sub''}{k''} $$
          \item $k_2 = k_1$ and $k'' = k'$, from which we deduce $k_2 + k = k_1 + k \overset{\textup{IH}}= k' = k''$.
          \item $\pc_2 = \pc_1$ and $\pc'' = \pc'$, so that $\sinterpret{\pc} \cap \sinterpretk\sub\idx^{-1}[\pc_2] = \sinterpret{\pc''}$ by IH.
          \item $\po_2 = \po_1 \sand {\sub_1(\bexp)}$ and $\po'' = \po' \sand {\sub'(\bexp)}$.
        \end{enumerate}
        In this case, the non-trivial item we have to verify is \rom 4, so we set out to prove that
        \begin{equation} \label{eq:canonicalobservegoal}
          \sinterpret {\po} \cap \sinterpretk\sub\idx^{-1}[\sinterpret{\po_2}] = \sinterpret{\po''}
        \end{equation}
        For this, we first observe that
        \begin{itemize}
          \item $\sinterpret{\po''} = \sinterpret{\po'} \cap \sinterpret{\sub'(b)}$;
          \item $\sinterpret{\po_2} = \sinterpret{\po_1} \cap \sinterpret{\sub_1(b)}$; and
          \item by \Cref{item:substcompbexp} and IH, we have $\sinterpretk\sub\idx^{-1}[\sinterpret{\sub_1(b)}] = \sinterpret{\sub'(b)}$.
          \item The IH also tells us that $\sinterpret{\po'} = \sinterpret{\po} \cap \sinterpretk\sub\idx^{-1}[\sinterpret{\po_1}]$.
        \end{itemize}
        We can now make the following derivation:
        $$ \begin{array}{rl}
          \sinterpret{\po} \cap \sinterpretk\sub\idx^{-1} [\sinterpret{\po_2}]
          = & \sinterpret{\po} \cap \sinterpretk\sub\idx^{-1}\big[\sinterpret{\po_1} \cap \sinterpret{\sub_1(\bexp)}\big] \\
          = & \sinterpret{\po} \cap \sinterpretk\sub\idx^{-1}\big[\sinterpret{\po_1}\big] \cap \sinterpretk\sub\idx^{-1}\big[\sinterpret{\sub_1(b)}\big] \\
          = & \sinterpret{\po'} \cap \sinterpretk\sub\idx^{-1}\big[ \sinterpret{\sub_1(b)} \big] \\
          = & \sinterpret{\po'} \cap \sinterpret{\sub'(b)} \\
          = & \sinterpret{\po''} \\
        \end{array}$$
        and we have thus concluded $\cfg'' \cfgeq \cfg_2 \cfgcmp \cfg$.

      \item In the case for \refRule{seq-0}, where $\r=\seq\pskip\q$, there is nothing to prove, because $\cfg_2 = \cfg_1$, and $\cfg'' = \cfg'$.
        Thus, $\cfg'' \cfgeq \cfg_2 \cfgcmp \cfg$ directly by IH $\cfg' \cfgeq \cfg_1 \cfgcmp \cfg$.

      \item \refRule{seq-n}: the only inductive step.
        Here $\q = \seq{\q_1}{\q_2}$ and $\r = \seq{\q'_1}{\q_2}$, and so
        $$  \srule{}{(\q_1,\cfg_1)\tr(\q_1',\cfg_2)}
                    {(\seq{\q_1}{\q_2},\cfg_1)\tr(\seq{\q_1'}{\q_2},\cfg_2)} 
      \qquad\srule{}{(\q_1,\cfg') \tr(\q_1',\cfg'')}
                    {(\seq{\q_1}{\q_2},\cfg') \tr(\seq{\q_1'}{\q_2},\cfg'')}$$
        Here, we may assume that the transitions from $\q_1$ to $\q_1'$ are not by \refRule{seq-n} or \refRule{seq-0}, by observing that sequencing is associative, so we can evaluate right-associatively, without loss of generality.
        The analyses from all other cases for the transition from $\q_1$ to $\q_1'$ (and using the chain $(\q_1,\initconfig)\trCl(\q_1,\initconfig)$ of length $\l=0$) can now be repeated.
        This will establish that $(\q_1,\initconfig)\tr(\q'_1,\initconfig')$ for some $\initconfig' = (\sub'_0,\idx'_0,\pc'_0,\po'_0)$ that satisfies (1) $\cfg_2 \cfgeq \initconfig' \cfgcmp \cfg_1$ and (2) $\cfg'' \cfgeq \initconfig' \cfgcmp \cfg'$.
        With this, we verify that $\cfg'' \cfgeq \cfg_2 \cfgcmp \cfg$.
    \begin{enumerate}[label=(\roman*)]
      \addtocounter{enumi}{1}
      \item $k_2 + k \overset{(1)}= k'_0 + k_1 + k \overset{\textup{IH}}= k'_0 + k' \overset{(2)}= k''$;
      \addtocounter{enumi}{-2}
    \item $\sinterpretk{\sub_2}{k_2} \circ \sinterpretk{\sub}{\idx} \overset{(1)}= \sinterpretk{\sub'_0}{k'_0} \circ \sinterpretk{\sub_1}{k_1} \circ \sinterpretk{\sub}{\idx} \overset{\textup{IH}}= \sinterpretk{\sub'_0}{k'_0} \circ \sinterpretk{\sub'}{k'} \overset{(2)}= \sinterpretk{\sub''}{k''}$
      \addtocounter{enumi}{1}
    \item We can make the following derivation of sets:
      $$ \begin{array}{rl}
        \sinterpret{\pc} \cap \sinterpretk{\sub}{\idx}^{-1}\big[\sinterpret{\pc_2}\big]
        \overset{(1)}= & \sinterpret{\pc} \cap \sinterpretk{\sub}{\idx}^{-1}\Big[\sinterpret{\pc_1} \cap \sinterpretk{\sub_1}{k_1}^{-1}\big[\sinterpret{\pc'_0}\big] \Big] \\
        \overset{\textup{IH}}= & \sinterpret{\pc'} \cap \sinterpretk{\sub}{\idx}^{-1}\big[\sinterpretk{\sub_1}{k_1}^{-1}[\sinterpret{\pc'_0}]\big] \\
        \overset{\textup{IH}}= & \sinterpret{\pc'} \cap \sinterpretk{\sub'}{k'}^{-1}\big[\sinterpret{\pc'_0}\big] \\
        \overset{(2)}= & \sinterpret{\pc''} \\
      \end{array}$$
      This concludes this item.
    \item The path observation $\po''$ is done in exactly the same way as the path condition.
  \end{enumerate}
  Conclude that for every $\p$ and $\cfg$, $(\p,\initconfig)\trCl(\r,\cfg_2)$ (of length $\l+1$) if and only if $(\p,\cfg)\trCl(\r,\cfg'')$ (of length $\l+1$) such that $\cfg''\cfgeq\cfg_2\cfgcmp\cfg$ in case the last transition rule of the chains was proven by \refRule{seq-n}.

      \item \refRule{if-T}: looking at the rule, we know that
        \begin{enumerate}[label=(\roman*)]
          \item and \rom 2: $k_2 = k_1$ and $k'' = k'$, and $\sub_2 = \sub_1$ and $\sub'' = \sub'$, so there is nothing to prove again.
            \addtocounter{enumi}{1}
          \item the interesting case, where $\pc_2 = \pc_1 \sand \sub_1(\bexp)$ and $\pc'' = \pc' \sand \sub'(\bexp)$.
          \item $\po_2 = \po_1$ and $\po'' = \po'$, so nothing to prove here.
        \end{enumerate}
        We check item \rom 3, and prove that
          $$ \sinterpret {\pc} \cap \sinterpretk\sub\idx^{-1}[\sinterpret{\pc_2}] = \sinterpret{\pc''} $$
        For this, we observe
        \begin{itemize}
          \item $\sinterpret{\pc''} = \sinterpret{\pc'} \cap \sinterpret{\sub'(b)}$;
          \item $\sinterpret{\pc_2} = \sinterpret{\pc_1} \cap \sinterpret{\sub_1(b)}$; and
          \item by \Cref{item:substcompbexp} and IH, $\sinterpretk\sub\idx^{-1}[\sinterpret{\sub_1(b)}] = \sinterpret{\sub'(b)}$.
          \item by IH, $\sinterpret{\pc'} = \sinterpret{\pc} \cap \sinterpretk\sub\idx^{-1}[\sinterpret{\pc_1}]$.
        \end{itemize}
        With this, we derive:
        $$ \begin{array}{rl}
          \sinterpret{\pc} \cap \sinterpretk\sub\idx^{-1} [\sinterpret{\pc_2}]
          = & \sinterpret{\pc} \cap \sinterpretk\sub\idx^{-1}\big[\sinterpret{\pc_1} \cap \sinterpret{\sub_1(\bexp)}\big] \\
          = & \sinterpret{\pc} \cap \sinterpretk\sub\idx^{-1}\big[\sinterpret{\pc_1}\big] \cap \sinterpretk\sub\idx^{-1}\big[\sinterpret{\sub_1(b)}\big] \\
          = & \sinterpret{\pc'} \cap \sinterpretk\sub\idx^{-1}\big[ \sinterpret{\sub_1(b)} \big] \\
          = & \sinterpret{\pc'} \cap \sinterpret{\sub'(b)} \\
          = & \sinterpret{\pc''} \\
        \end{array}$$
        and we have thus concluded $\cfg'' \cfgeq \cfg_2 \cfgcmp \cfg$ in this case.

      \item \refRule{if-F} is analogous to \refRule{if-T}, where we just take complements.

      \item \refRule{iter-F}: we have that
        \begin{enumerate}[label=(\roman*)]
          \item and \rom 2 $k_2 = k_1$ and $k'' = k'$, and $\sub_2 = \sub_1$ and $\sub'' = \sub'$; nothing to prove,
            \addtocounter{enumi}{1}
          \item $\pc_2 = \pc_1 \sand \sub_1(\pneg \bexp)$ and $\pc'' = \pc' \sand \sub'(\pneg \bexp)$; the interesting case, and
          \item $\po_2 = \po_1$ and $\po'' = \po'$, so nothing to prove.
        \end{enumerate}
        This case is again concluded analogously to \refRule{if-T}.

      \item \refRule{iter-T}: perhaps surprisingly, this case is exactly the same as \refRule{if-T}.
        This is because we are doing forward induction on the length of the transition chain, so the symbolic execution is inherently finite.
        Indeed, it follows from the fact that $\cfg_2 = \cfg_1$ and $\cfg''=\cfg'$, with the exception of the path conditions, that checking item \rom 3 is the only relevant task, and this is done exactly as in \refRule{if-T}.

    \end{itemize}
    This finishes induction on the rules in \Cref{fig:symos}, and concludes the proof of \Cref{lem:canonicalexecution}.
  \end{proof}

  \begin{proposition}[Surjectivity of the Bijection] \label{prop:mappingsurjective}
    If $(F,B,\OO) \in \FF_{\p}$ then there is a configuration $\cfg$ such that $(\p,\initconfig) \trCl (\pskip,\cfg)$ and \rom{1} $\sinterpretk{\sub}{\idx} = F$, \rom{2} $\sinterpret {\pc} = B$, and \rom{3} $\sinterpret{\po} = \OO$.
  \end{proposition}
  \begin{proof}
    By induction on the structure of $\p$.
    \begin{itemize}
      \item $\p =\pskip$: we have $(F,B,\OO) = (\id_{\statespace},\statespace,\statespace)$ and, indeed, $(\pskip,\cfg_0) \trCl (\pskip,\cfg_0)$ where $\cfg_0$ satisfies \rom 1-\rom 3;
      \item $\p =\assign$: we have $(F,B,\OO) = (\rho \mapsto \rho\update{i}{\interpret{\exp}(\res \rho n)},\statespace,\statespace)$ and $(S,\cfg_0) \trCl (\pskip,\cfg)$ with $\cfg = (\initsub\update{i}{\initsub(\exp)},0,\strue,\strue)$ and one verifies in exactly the same way as in the proof of \Cref{lem:symbolicmappingsmallstep} that $\sinterpretk{\sub}{\idx} = F$ in this case.
      \item Idem for $\p =\psample$, where $(F,B,\OO) = (\rho \mapsto \sample i (\rho),\statespace,\statespace)$ and $(S,\initconfig) \trCl (\pskip,\cfg)$ with $\cfg = (\initsub\update{i}{\y_0},1,\strue,\strue)$.
      \item Now if $\p = \observe$, then $(F,B,\OO) = (\id_{\statespace},\statespace,\bexpset)$ and $\cfg = (\initsub,0,\strue,\app\initsub\bexp)$.
        Again, checking that these data satisfy \rom 1-\rom 3 is done in the same way as in the proof of \Cref{lem:symbolicmappingsmallstep}.
    \item Consider now the case where $\p = \seq {\p_1}{\p_2}$, where we have two IHs, namely one for all $(F_1,B_1,\OO_1) \in \FF_{\p_1}$ and the other for all $(F_2,B_2,\OO_2) \in \FF_{\p_2}$.
      So let $(F,B,\OO) \in \FF_{\p} = \FF_{\seq{\p_1}{\p_2}}$.
      Then, by definition, there are $(F_1,B_1,\OO_1) \in \FF_{\p_1}$ and $(F_2,B_2,\OO_2) \in \FF_{\p_2}$, such that
      \begin{itemize}
        \item $F = F_2 \circ F_1$;
        \item $B = B_1 \cap F_1^{-1}[B_2]$; and
        \item $\OO = \OO_1 \cap F_1^{-1}[\OO_2]$.
      \end{itemize}
      and by the two IHs, we obtain two canonical symbolic executions
      $$ (\p_1,\initconfig) \trCl (\pskip,\cfg_1) \qquad \textup{and} \qquad (\p_2,\initconfig) \trCl (\pskip,\cfg_2), \quad (*) $$
      By repeatedly applying \refRule{seq-n} and finally \refRule{seq-0}, we then also have
      $$ (\seq {\p_1}{\p_2},\initconfig) \trCl (\seq{\pskip}{\p_2},\cfg_1) \tr (\p_2,\cfg_1) \trCl (\pskip,\cfg) $$
      for some $\cfg$.
      Then, using the canonical symbolic transition chain (*), and by \Cref{lem:canonicalexecution}, we have $\cfg \cfgeq \cfg_2 \cfgcmp \cfg_1$.
      We conclude this case by showing that properties \rom 1-\rom 3 are satisfied by exactly this $\cfg$:
      \begin{enumerate}[label=(\roman*)]
        \item $\sinterpretk{\sub}{\idx} = \sinterpretk{\sub_2}{k_2} \circ \sinterpretk{\sub_1}{k_1} = F_2 \circ F_1 = F$;
        \item $\sinterpret{\pc} = \sinterpret{\pc_1} \cap \sinterpretk{\sub_1}{k_1}^{-1}[\pc_2] = B_1 \cap F_1^{-1}[B_2] = B$; and
        \item $\sinterpret{\po} = \sinterpret{\po_1} \cap \sinterpretk{\sub_1}{k_1}^{-1}[\po_2] = \OO_1 \cap F_1^{-1}[\OO_2] = \OO$.
      \end{enumerate}
      And so the case for sequencing is finished.
      \item Now let $\p = \pif b{\p_1}{\p_2}$, and $(F,B',\OO) \in \FF_{\p}$.
        This can happen in either of two ways:
        \begin{enumerate}
          \item $(F,B,O) \in \FF_{\p_1}$ and $B' = B \cap (\bexpset)$.
            In this case, by the IH for $\p_1$, we have $\cfg_1$ such that $ (\p_1,\cfg_0) \trCl (\pskip,\cfg_1)$ with the required properties \rom 1-\rom 3.
            This is the canonical execution, so we have $(\pif b{\p_1}{\p_2},\cfg_0) \tr (\p_1,\cfg') \trCl (\pskip,\cfg'')$ for some $\cfg'$ with $\cfg'' \cfgeq \cfg_1 \cfgcmp \cfg'$ by \Cref{lem:canonicalexecution}.
            But we know $\sub' = \initsub$, $k'=0$, $\pc' = \strue \sand \app\initsub\bexp$, and $\po' = \strue$.
            With this we verify \rom 1-\rom 3 for $\cfg''$:
            \begin{enumerate}[label=(\roman*)]
              \item $\sinterpretk{\sub''}{k''} = \sinterpretk{\sub_1}{k_1} \circ \sinterpretk{\sub'}{k'} = \sinterpretk{\sub_1}{k_1} \circ \id_{\statespace} = \sinterpretk{\sub_1}{k_1}$, and by IH for $\p_1$, this equals $F$.
              \item $\sinterpret{\pc''} = \sinterpret{\pc'} \cap \sinterpretk{\sub'}{k'}^{-1} \big[{\sinterpret{\pc_1}}\big] = \sinterpret{\strue \sand \app\initsub\bexp} \cap \id_{\statespace}^{-1}\big[\sinterpret{\pc_1}\big]$, and by \Cref{lem:substitutionlemmabexp} and the IH that $\sinterpret{\pc_1}=B$, this is $(\bexpset) \cap B = B' $.
              \item $\sinterpret{\po''} = \sinterpret{\po'} \cap \id_{\statespace}^{-1}\big[ \sinterpret{\strue} \big] = \sinterpret{\po'} = \OO$.
            \end{enumerate}
            Thus, in the case that $(F,B,O) \in \FF_{\p_1}$ and $B' = B \cap (\bexpset)$, we have
            $$ (S,\cfg_0) \trCl (\pskip,\cfg'') $$
            where $\cfg''$ satisfies the properties \rom 1-\rom 3, so we are done.
          \item The case that $(F,B,O) \in \FF_{\p_2}$ and $B' = B \cap (\bexpcset)$ is analogous; just replace $\pc' = \strue \sand \app\initsub\bexp$ by $\pc' = \strue \sand \app\initsub{\pneg\bexp}$.
        \end{enumerate}
        We have thus finished the inductive step in case $\p = \pif b{\p_1}{\p_2}$.

      \item We conclude the proof with iteration; let $\p = \while\bexp\q$.
        Then, for $(F,B,\OO) \in \FF_{\p}$ means there is $m$ such that $(F,B,O) \in \F^m_{\bexp,\q} (\F_0)$ (where $\F_0 = \{(\id_{\statespace},\bexpcset,\statespace)\}$).
        We therefore proceed by induction on $m$ that the following holds:
        for all $(F,B,\OO) \in \F^m_{\bexp,\q}(\F_0)$, there is $\cfg$ such that $(\while\bexp\q,\cfg_0\trCl(\pskip,\cfg)$.

        The base case is when $m = 0$, in which case we have $(F,B,\OO) = (\id_{\statespace},\bexpcset,\statespace)$, and, indeed,
        $ (\while\bexp\q,\initconfig)\tr(\pskip,\initsub,0,\strue \sand \app\initsub{\pneg\bexp}, \strue$.
        In this setting, it is routine to verify that properties \rom 1-\rom 3 hold.

        We now prove the statement for $m+1$.
        $(F,B,\OO) \in \F^{m+1}_{\bexp,\q} (\F_0)$ means there are $(F_\q,B_\q,\OO_\q) \in \FF_\q$ and $(F_m,B_m,\OO_m) \in \F^m_{\bexp,\q} (\F_0)$, such that
        \begin{itemize}
          \item $F = F_m \circ F_\q$;
          \item $B = (\bexpset) \cap B_\q \cap F_\q^{-1}[B_m]$; and
          \item $\OO = \OO_\q \cap F_\q^{-1}[\OO_m]$.
        \end{itemize}
        Now apply the IH for the subterm $\q$ of $\p$ for $(F_\q,B_\q,\OO_\q) \in \FF_\q$, and for the case $m$ in the induction on $m$ for $(F_m,B_m,\OO_m) \in \F^m_{\bexp,\q}(\F_0)$, to obtain \emph{canonical} symbolic executions
        $$ (\q,\initconfig)\trCl(\pskip,\cfg_\q) \quad (*) 
        \qquad \textup{and} \quad 
        (\while\bexp\q,\initconfig)\trCl(\pskip,\cfg_m) \quad (**) $$
        where $\cfg_\q = (\sub_\q,k_\q,\pc_\q,\po_\q)$ and $\cfg_m = (\sub_m,\idx_m,\pc_m,\po_m)$ with the properties that
        \begin{itemize}
          \item $\sinterpretk{\sub_\q}{k_\q} = F_\q$ and $\sinterpretk{\sub_m}{k_m} = F_m$;
          \item $\sinterpret{\pc_\q} = B_\q$ and $\sinterpret{\pc_m} = B_m$; and
          \item $\sinterpret{\po_\q} = \OO_\q$ and $\sinterpret{\po_m} = \OO_m$.
        \end{itemize}
        Using \Cref{lem:canonicalexecution} twice, we have an outgoing transition chain from $(\while\bexp\q,\initconfig)$ as follows:
        $$ \begin{array} {rll}
          (\while\bexp\q,\initconfig){}
          \tr&(\seq\q{\while\bexp\q},\initconfig') &\textup{\refRule{iter-T}} \\
          \trCl&(\seq\pskip{\while\bexp\q},\cfg')  
                          &\textup{\refRule{seq-n} and \Cref{lem:canonicalexecution} with (*)} \\
          \tr&(\while\bexp\q,\cfg')                &\textup{\refRule{seq-0}} \\
          \trCl&(\pskip,\cfg)                      &\textup{\Cref{lem:canonicalexecution} with (**)}
      \end{array}$$
      where,  $\initconfig' = (\initsub,0,\strue \sand \initsub(b),\strue)$, and (1) $\cfg' \cfgeq \cfg_\q \cfgcmp \initconfig'$ and (2) $\cfg \cfgeq \cfg_m \cfgcmp \cfg'$.
      We verify the properties \rom 1-\rom 3 for this $\cfg$:
      \begin{enumerate}[label=(\roman*)]
        \item $\sinterpretk\sub\idx \overset{(2)}= \sinterpretk{\sub_m}{k_m} \circ \sinterpretk{\sub'}{k'} \overset{(1)}= \sinterpretk{\sub_m}{k_m} \circ \sinterpretk{\sub_\q}{k_\q} \circ \sinterpretk {\initsub}0 \overset{\textup{IH}}= F_m \circ F_\q \circ \id_{\statespace} = F$;
        \item $\sinterpret{\pc} \overset{(2)}= \sinterpret{\pc'} \cap \sinterpretk{\sub'}{k'}^{-1}[\sinterpret{\pc_m}] \overset{(1)}= \sinterpret{\strue \sand \initsub(b)} \cap \sinterpretk{\initsub}0^{-1}[\sinterpret{\pc_\q}] \cap (\sinterpretk {\sub_\q}{k_\q} \circ \sinterpretk{\initsub}{0})^{-1}[B_m]$.
          It is routine to equate this to $(\bexpset) \cap B_\q \cap F_\q^{-1}[B_m] = B$.
        \item Analogous to the above property
      \end{enumerate}
      We have thus verified that there is $\cfg$ such that $(\while bT,\initconfig) \trCl (\pskip,\cfg)$ for $(F,B,\OO) \in \F^{m+1}_{b,T}(\F_0)$ such that $\sinterpretk\sub\idx = F$, $\sinterpret{\pc} = B$, and $\sinterpret{\po}=\OO$.
    \end{itemize}
    Induction on $\p$ has finished, and we finished the proof of \Cref{prop:mappingsurjective}.
    A notable consequence of this is that the mappings $\Phi_\p$ defined in the proof of \Cref{thm:symbexsemantics} are surjective.
  \end{proof}

  \begin{proposition}[Injectivity of the Bijection]\label{prop:mappinginjective}
    Let $\p \in \P$.
    If $(\p,\initconfig) \trCl (\pskip,\cfg_1)$ and $(\p,\initconfig) \trCl (\pskip,\cfg_2)$
    then either $\cfg_1=\cfg_2$ or $\sinterpret{\pc_1} \cap \sinterpret{\pc_2} = \emptyset$.
  \end{proposition}

\subsection{Proof of \Cref{thm:mainresult}} \label{app:mainresultproof}
We restate the theorem for convenience:
\mainresult*
\begin{proof}
For every program $\p$, there is a function $f_\p$ such that
\begin{enumerate}[label=(\roman*)]
  \item\label{item:osdencorrespondence} $(\mu\otimes\infleb) \Big(\f^{-1}_{\p}[A\times\R^{\omega}]\Big) = \sem \p (\mu) (A)$, and
  \item\label{item:ossymcorrespondence} $(\mu\otimes\infleb) \Big(\f^{-1}_{\p}[A\times\R^{\omega}]\Big) = \sum_{(F,B,\OO) \in \FF_\p} (\mu\otimes\infleb)(F^{-1}[A\times\R^{\omega}] \cap B \cap \OO)  $.
\end{enumerate}
The function $f_\p$ is given in \Cref{def:osfun} below.
\Cref{item:osdencorrespondence} is \Cref{lem:osdencorrespondence} (take $B=A\times\R^\omega$).
\Cref{item:ossymcorrespondence} follows from \Cref{lem:ossymcorrespondence}, which gives
$$ \f_\p^{-1}[A\times\R^\omega] = \bigcup_{(F,B,\OO)\in\FF_\p} F^{-1}[A\times\R^\omega] \cap B \cap \OO, $$
and measuring the union on the right-hand side is the same as summing all the individual measures -- as in \Cref{item:ossymcorrespondence} -- because it is a disjoint union by \Cref{lem:disjointpathconditions}.
The theorem is proven by chaining the two equalities in \Cref{item:osdencorrespondence,item:ossymcorrespondence} and applying \Cref{thm:symbexsemantics}. 
\par\hfill Q.E.D.
\end{proof}

\begin{definition}[Operational semantics \cite{kozen1979semantics}, extended to support observe statements] \label{def:osfun}
  For programs $\p \in \P$ in $n$ variables, the partial functions $f_\p : \R^{n+\omega}\rightharpoonup \R^{n+\omega} \cup \{\abort\}$ are defined inductively on the structure of $\p$ as:
    $$\f_\p : \rho \mapsto
    \begin{cases}
      \rho & \textup{if } \p = \pskip \\
      \rho \update {i} {\interpret {\exp} (\res{\rho}{n})}
                         & \textup{if } \p = \assign \\
      \sample i (\rho)
                         & \textup{if } \p = \psample \\
      \rho & \textup{if } \p = \observe \textup { and } \res \rho n \in \interpret {\bexp} \\
      \abort & \textup{if } \p = \observe \textup { and } \res \rho n \not \in \interpret {\bexp} \\
      \f_\p (\rho)
                         & \textup{if } \p = \pif b {\p_1}{\p_2} \textup { and } \res \rho n \in \interpret{\bexp}  \\
      \f_{\p_2} (\rho)
                         & \textup{if } \p = \pif b {\p_1}{\p_2} \textup { and } \res \rho n \not\in \interpret{\bexp}  \\
      \f_\q^m (\rho)
                         & \textup{if } \p = \while b {\q}\textup{, } m := \min_{} \{ j \in \N \setst  \res {\f_\q^j (\rho)} n \not \in \interpret{\bexp} \} \\
      (\f_{\p_2} \circ \f_{\p_1}) (\rho)
                         & \textup{if } \p = \seq{\p_1}{\p_2} \\
    \end{cases}
    $$
    where $\f^m$ denotes $m$-fold iterated applications of $\f$ (identity for $m=0$).
\end{definition}
In this definition, $\abort$ indicates an aborted execution due to an observe statement.
The functions may be partial because loops might diverge:
for certain $\rho \in \statespace$ there may not be $j \in \N$ for
which
$\f_\q^j (\rho) \not \in \interpret{\bexp} \times
\R^\omega$; we write $\f_\p(\rho) \diverges$ if $\p$ diverges on input $\rho$.
A failed observe yields the explicit \emph{aborted} state
$\abort$ (not the same as divergence!) to which the domain naturally extends.

The distribution of outputs through this operational semantics is the
(possibly unnormalized) one defined by the denotational semantics:
\begin{lemma}[Correctness of operational semantics \cite{kozen1979semantics}, extended to support observe statements] \label{lem:osdencorrespondence} 
    Let $\p\in\P$ be a program in $n$ variables and $\mu \in \ms{\R^n}$ a probability measure over the input variables.
    Then, for every $B \in \BB(\R^{n+\omega})$,
    \begin{equation}
      \label{eq:kozengoal}
      (\sem\p(\mu) \otimes \infleb) (B) = (\mu\otimes\infleb) (\f_{\p}^{-1}[B])
    \end{equation}
\end{lemma}
\begin{proof}
    For arbitrary $\p\in\P$, define
    $$ \HH_\p := \{ B \in \BB(\R^{n+\omega}) \setst (\sem\p(\mu) \otimes \infleb) (B) = (\mu\otimes\infleb) (\f_{\p}^{-1}[B]) \textup{ for all } \mu \in \ms{\R^n} \} $$
    Assume now that $\HH_\p$ contains all the cylinder sets for every $\p$.
    It is not hard to see that every $\HH_\p$ is also a Dynkin system:
    \begin{itemize}
      \item $\HH_\p$ always contains the emptyset.
      \item If $B \in \HH_\p$, then
        $$ (\sem\p(\mu) \otimes \infleb)(\compl B) = (\sem\p(\mu) \otimes \infleb) (\R^{n+\omega}) - (\sem\p(\mu) \otimes \infleb) (B) $$
        and, since $\R^{n+\omega} \in \HH_\p$ as a cylinder set, and $B \in \HH_\p$, this is equal to
        $$ (\mu\otimes\infleb) (\f_{\p}^{-1} [\R^{n+\omega}]) - (\mu\otimes\infleb) (\f_{\p}^{-1}[B]) = (\mu\otimes\infleb) (\f_{\p}^{-1} [\compl B]), $$
        so $\HH_\p$ is always closed under complement.
      \item Both $(\sem\p(\mu) \otimes \infleb)$ and $\pushforward {(\mu\otimes\infleb)} {\f_{\p}}$ are measures, so $\HH_\p$ is obviously closed under countable unions of pairwise disjoint sets.
    \end{itemize}
    By the $\pi$-$\lambda$-theorem, $\HH_\p$ therefore contains all Borel sets.

    It remains to show the assumption that $\HH_\p$ contains the cylinder sets for every $\p$.
    Thus, we show \eqref{eq:kozengoal} by induction on the structure of $\p$, where we may assume that $B$ is a cylinder set, i.e., it can be written as
    $$ B = \underbrace{B_0 \times \dots \times B_{n-1}}_{=:B^{(n)}} \times \underbrace{B_n \times B_{n+1} \times \dots \times B_{n+p} \times \R^\omega}_{=:B^{(\omega)}} $$
    for some $p \in \N$.
    \begin{itemize}
    \item For $\p = \pskip$, we have $\sem\p(\mu) = \mu$, so the LHS of \eqref{eq:kozengoal} is just $(\mu\otimes\infleb) (B)$ for any measurable set $B$.
      Also $\f_{\p}^{-1}[B] = B$, so the RHS equals $(\mu\otimes\infleb) (B)$ as well.
    \item For $\p = \assign$, for the LHS \eqref{eq:kozengoal} we have $\sem\p(\mu) = \pushforward{\mu}{\asgn^i_\exp}$, where
      $$\asgn_{\exp}^i(x_0,\ldots,x_{n-1}) = (x_0, \ldots, x_{i-1}, \interpret{\exp}(x_0,\ldots,x_{n-1}), x_{i+1}, \ldots, x_{n-1} ) $$
      For this $\asgn_{\exp}^i$, we can compute the following preimage:
      \begin{equation} \label{eq:kozenassign}
        (\asgn_{\exp}^i)^{-1}[B^{(n)}] = (B_0 \times \dots \times B_{i-1} \times \R \times B_{i+1} \times \dots \times B_{n-1}) \cap \interpret{\exp}^{-1}[B^{(n)}]
      \end{equation}
      Hence, the LHS of \eqref{eq:kozengoal}, which is $(\sem\p(\mu) \otimes \infleb) (B^{(n)} \times B^{(\omega)})$, becomes the $\mu$ of this set \eqref{eq:kozenassign}, multiplied by $\infleb (B^{(\omega)})$.
      For the RHS, the preimage of $B$ under $\f_{\p} : \rho \mapsto \rho\update{i}{\interpret{\exp}(\res \rho n)} $ is
      $$\f_{\p}^{-1}[B] = \Big( (B_0 \times \dots \times B_{i-1} \times \R \times B_{i+1} \times \dots \times B_{n-1}) \cap \interpret{\exp}^{-1}[B^{(n)}] \Big) \times B^{(\omega)}, $$
      and measuring this with $(\mu\otimes\infleb)$ establishes this case.
    \item Now let $\p = \psample$.
      We have, on the LHS of \eqref{eq:kozengoal},
      $$ \sem\p(\mu) (B_0 \times \dots \times B_{n-1}) = \mu (B_0 \times \dots \times B_{i-1} \times \R \times B_{i+1} \times \dots \times B_{n-1}) \cdot \lambda (B_i) $$
      and $(\sem\p(\mu) \otimes \infleb) (B)$ is exactly this, multiplied by $\infleb (B^{(\omega)})$.

      For the RHS, $\f_{\p} (\rho) = \sample i (\rho)$, where
      $$ \sample i (\rho) = (\rho_0, \ldots, \rho_{i-1}, \bm{\rho_n}, \rho_{i+1}, \ldots, \rho_{n-1}, \bm{\rho_{n+1}}, \bm{\rho_{n+2}}, \ldots ) $$
      so that
      $$\f_{\p}^{-1} [B] = B_0 \times \dots \times B_{i-1} \times \R \times B_{i+1} \times \dots \times B_{n-1} \times B_i \times B^{(\omega)} $$
      and measuring this with $(\mu\otimes\infleb)$ establishes the case.
    \item The last base case is $\p = \observe$.
      We have, for the LHS,
      $$ (\sem\p(\mu) \otimes \infleb) (B) = \mu(B^{(n)} \cap \interpret{\bexp}) \cdot \infleb (B^{(\omega)}) $$
      On the RHS, $\f_{\p} (\rho) = \rho$ if $\rho \in \interpret{\bexp} \times \R^\omega$ (and undefined otherwise), so $\f_{\p}^{-1}[B] = (B^{(n)} \cap \interpret {\bexp}) \times B^{\omega}$.
      Measuring this last set with $(\mu\otimes\infleb)$ establishes this case.
    \end{itemize}

  The base cases of the structure of $\p$ are established, and we continue with the induction steps.
  The induction hypotheses for substatements $\p'$ of $\p$ (these are $\p_1$ and $\p_2$ in the case for \emph{if} and \emph{sequencing}, and $\q$ in the case for \emph{while}) are that $\HH_{\p'}$ contains all cylinder sets.
  By the reasoning above, the IH therefore entails that $\HH_{\p'}$ moreover contains all Borel sets.
  \begin{itemize}
    \item If $\p = \pif \bexp{\p_1}{\p_2}$, we have two induction hypotheses (IH), which we apply respectively using the measures $\restrictMeasure{\interpret{\bexp}}(\mu)$ and $\restrictMeasure{\interpret{\pneg \bexp}}(\mu)$ in the following derivation:
      $$ \begin{array}{l}
        (\sem\p(\mu) \otimes \infleb)(B) \\
        \quad = (\sem\p(\mu) \otimes \infleb)(B^{(n)} \times B^{(\omega)}) \\
        \quad = \sem\p(\mu)(B^{(n)}) \cdot \infleb(B^{(\omega)}) \\
        \quad = \Big(\sem{\p_1}(\restrictMeasure{\interpret{\bexp}}(\mu))(B^{(n)}) + \sem{\p_2}(\restrictMeasure{\interpret{\pneg \bexp}}(\mu))(B^{(n)}) \Big) \cdot \infleb (B^{(\omega)}) \\
        \quad = \sem{\p_1}(\restrictMeasure{\interpret{\bexp}}(\mu))(B^{(n)}) \cdot \infleb (B^{(\omega)}) + \sem{\p_2}(\restrictMeasure{\interpret{\pneg \bexp}}(\mu))(B^{(n)}) \cdot \infleb (B^{(\omega)}) \\
        \quad = (\sem{\p_1}(\restrictMeasure{\interpret{\bexp}}(\mu)) \otimes \infleb)(B) + (\sem{\p_2}(\restrictMeasure{\interpret{\pneg \bexp}}(\mu)) \otimes \infleb) (B) \\
        \quad \overset{\textup{IH}}= (\restrictMeasure{\interpret{\bexp}}(\mu) \otimes \infleb) (\f_{\p_1}^{-1}[B]) + (\restrictMeasure{\interpret{\pneg \bexp}}(\mu)\otimes \infleb) (\f_{\p_2}^{-1}[B]) \\
        \quad \overset{\textup{(*)}}= (\mu\otimes\infleb) \big(\f_{\p_1}^{-1} [B] \cap (\interpret{\bexp} \times \R^\omega)\big) + (\mu\otimes\infleb) \big(\f_{\p_2}^{-1} [B] \cap (\compl{\interpret{\bexp}} \times \R^\omega)\big) \\
        \quad = (\mu\otimes\infleb) \Big( \big(\f_{\p_1}^{-1} [B] \cap (\interpret{\bexp} \times \R^\omega)\big) \cup \big(\f_{\p_2}^{-1} [B] \cap \compl{\interpret{\bexp}} \times \R^\omega) \big) \Big)\\
        \quad = (\mu\otimes\infleb) (\f_{\p}^{-1}[B])
      \end{array}$$
      In the last equality we apply the definition of $\f_{\p}$.
      In the second to last equality, the sum of the measures is the measure of the union, because the sets are disjoint.
      Indeed, one is contained in $\interpret{\bexp} \times \R^\omega$ and the other in $\compl {\interpret{\bexp}} \times \R^\omega$.
      At (*), we use the fact that if $(X, \Sigma_X, \mu_X)$ and $(Y, \Sigma_Y, \mu_Y)$ are $\sub$-finite measure spaces and $A \in \Sigma_X$, then $(\restrictMeasure{A}(\mu_X) \otimes \mu_Y) (E) = (\mu_X \otimes \mu_Y) (E \cap (A \times Y))$ for all $E \in \Sigma_X \otimes \Sigma_Y$.
      (Here, $\Sigma_X = \BB(\R^n)$ and $\Sigma_Y = \BB(\R^\omega)$.)
      This case is now finished.
    \item Now let $\p = \while \bexp\q$.
      The LHS of \eqref{eq:kozengoal} is
      \begin{equation}
        \label{eq:kozenwhileLHS}
        \begin{array}{rl}
        (\sem\p(\mu) \otimes \infleb)(B)
        = & \sem\p(\mu)(B^{(n)}) \cdot \infleb(B^{(\omega)}) \\
        = & \Big( \sum_{m=0}^\infty \restrictMeasure{\interpret{\pneg\bexp}} \circ \big( \sem{\q} \circ \restrictMeasure{\interpret{\bexp}} \big)^m \Big) (\mu) (B^{(n)}) \cdot \infleb(B^{(\omega)}) \\
        = & \sum_{m=0}^\infty \restrictMeasure {\interpret{\pneg\bexp}} \big( (\sem{\q} \circ \restrictMeasure {\interpret{\bexp}} )^m (\mu) \big) (B^{(n)}) \cdot \infleb(B^{(\omega)}) \\
        = & \sum_{m=0}^\infty (\restrictMeasure {\interpret{\pneg\bexp}} \big( (\sem{\q} \circ \restrictMeasure {\interpret{\bexp}} )^m (\mu) \big) \otimes \infleb) (B) \\
        = & \sum_{m=0}^\infty ((\sem{\q} \circ \restrictMeasure {\interpret{\bexp}} )^m (\mu)\otimes \infleb) (B \cap (\compl{\interpret{\bexp}} \times \R^{\omega}))
      \end{array}
      \end{equation}
      Looking at the RHS, we derive the following expression:
      \begin{equation}
        \label{eq:kozenwhileRHS}
        \begin{array}{rl}
        \f_{\p}^{-1} [B]
        = & \{ \rho \in \R^{n+\omega} \setst \f_{\p}(\rho) \in B \} \\
        = & \{ \rho \in \R^{n+\omega} \setst \exists m ~.~ \f_\q^m(\rho) \in B \cap (\compl{\interpret{\bexp}} \times \R^\omega) \\
          & \qquad \wedge \forall j \in \{0,\ldots,m-1\} ~.~ \f_\q^j (\rho) \in \interpret{\bexp}\times \R^\omega \} \\
        = & \{ \rho \in \R^{n+\omega} \setst \exists m ~.~ \rho \in \f_\q^{-m}[B \cap (\compl{\interpret{\bexp}} \times \R^\omega)] \\
          & \qquad \wedge \forall j \in \{0,\ldots,m-1\} ~.~ \rho \in \f_\q^{-j} [\interpret{\bexp} \times \R^\omega] \} \\
          = & \bigcup_{m=0}^{\infty} \underbrace{\Big\{ \f_\q^{-m} [B \cap (\compl {\interpret{\bexp}} \times \R^{\omega})] \cap \bigcap_{j=0}^{m-1} \f_\q^{-j} [\interpret{\bexp} \times \R^{\omega}] \Big\}}_{=:U_m}
      \end{array}
      \end{equation}
      Here, $\f_\q^{-m}[A]$ denotes the $m$-th preimage, i.e., $\f_\q^{-m}[A] = A$ for $m=0$ and
      $$ \f_\q^{-m} [A] = \underbrace {\f_\q^{-1}[ \f_\q^{-1} [ \dots [}_{m \textup{ times}} A]\dots] $$
      The family of sets $(U_m)_{m \in \N}$ in \eqref{eq:kozenwhileRHS} are pairwise disjoint. That is, if $m \neq m'$ then $U_m \cap U_{m'} = \emptyset$.
      Indeed, without loss of generality, $m' > m$, and then
      $$ \rho \in U_{m'} \implies \rho \in \f_\q^{-m}[\interpret{\bexp}\times \R^\omega] \implies \rho \not \in \f_\q^{-m} [\compl {\interpret{\bexp}} \times \R^\omega] \implies \rho \not \in U_m $$
      Hence, for the RHS, we have
      $$ (\mu\otimes\infleb) (\f_{\p}^{-1}[B]) = (\mu\otimes\infleb) {\Big(\bigcup_{m=0}^\infty U_m\Big)} = \sum_{m=0}^\infty (\mu\otimes\infleb)(U_m) $$
      Comparing this with the derivation of the LHS \eqref{eq:kozenwhileLHS}, it will suffice to show the following holds for every $m$ and every $C \in \BB(\R^{n+\omega})$:
      \begin{equation}
        \label{eq:kozenwhilegoal} \big((\sem{\q} \circ \restrictMeasure {\interpret{\bexp}} )^m (\mu)\otimes \infleb\big) (C) = (\mu\otimes\infleb) (\f_\q^{-m} [C] \cap \bigcap_{j=0}^{m-1} \f_\q^{-j} [\interpret{\bexp} \times \R^{\omega}])
      \end{equation}
      Indeed, using $C = B \cap (\compl {\interpret{\bexp}} \times \R^{\omega})$ in \eqref{eq:kozenwhilegoal} will establish the result for this case.

      The claim \eqref{eq:kozenwhilegoal} can be proved by induction on $m$.
      \begin{itemize}
        \item The base case $m=0$ is easy, since both sides are obviously $(\mu\otimes\infleb) (C)$.
        \item For the inductive step, we use the IH (*) for the substatement $\q$ in $\p=\while \bexp\q$, and another IH (**) for $m=\l$.
          We prove the statement for $m=\l+1$.
          $$ \begin{array}{l}
            \big( (\sem{\q} \circ \restrictMeasure{\interpret{\bexp}})^{\l+1} (\mu) \otimes \infleb \big) (C) \\
            \quad = (\sem \q \big( \restrictMeasure{\interpret{\bexp}} \big( (\sem \q \circ \restrictMeasure{\interpret{\bexp}} )^\l(\mu) \big) \big) \otimes \infleb) (C) \\
            \quad \overset{\textup{*}}= (\restrictMeasure{\interpret{\bexp}} \big( (\sem \q \circ \restrictMeasure{\interpret{\bexp}})^\l (\mu)\big) \otimes \infleb) (\f_\q^{-1}[C]) \\
            \quad = \big( ( \sem \q \circ \restrictMeasure{\interpret{\bexp}} )^\l (\mu) \otimes \infleb \big) (\f_\q^{-1}[C] \cap (\interpret{\bexp}\times \R^\omega)) \\
            \quad \overset{\textup{**}}= (\mu\otimes\infleb) (\f_\q^{-\l} [\f_\q^{-1} [C] \cap (\interpret{\bexp} \times \R^\omega)] \cap \bigcap_{j=0}^{\l-1} \f_\q^{-j}[\interpret{\bexp}\times\R^\omega]) \\
            \quad = (\mu\otimes\infleb) (\f_\q^{-(\l+1)}[C] \cap \bigcap_{j=0}^{\l} \f_\q^{-j}[\interpret{\bexp}\times\R^\omega]) \\
          \end{array}$$
      \end{itemize}
      The inductive proof shows that \eqref{eq:kozenwhilegoal} holds for all $m$ and every $C \in \BB(\R^{n+\omega})$.

      Conclude that $(\sem\p(\mu) \otimes \infleb)(B) = (\mu\otimes\infleb) (\f_{\p}^{-1}[B])$ for all $B \in \BB(\R^{n+\omega})$ and all $\mu \in \ms{\R^n}$ in the case that $\p = \while bT$.
    \item For the case where $\p = \seq {\p_1} {\p_2}$, we have
      $$ (\sem\p(\mu) \otimes \infleb) (B) = (\sem{\p_2 \otimes \infleb)(\sem{\p_1}(\mu))}(B) $$
      Now we apply the IH for $\p_2$ with $\sem{\p_1}(\mu)$ and then for $\p_1$ with $\mu$, to rewrite this further as
      $$ \ldots = (\sem{\p_1}(\mu) \otimes \infleb) (\f_{\p_2}^{-1}[B]) = (\mu\otimes\infleb) (\f_{\p_1}^{-1} [\f_{\p_2}^{-1}[B]]) = (\mu\otimes\infleb) (\f_{\p}^{-1}[B]) $$
      Here, we again use the crucial fact that $\HH_{\p_1}$ is a Dynkin systems that contains the cylinder sets, so they contain the Borel sets.
      In particular, $\HH_{\p_1}$ contains $\f_{\p_2}^{-1}[B]$.
      This concludes the case for sequencing.
  \end{itemize}
  We have covered all cases for $\p$, and the proof is done.
  \end{proof}

\begin{lemma}[Correspondence between symbolic and concrete semantics] \label{lem:ossymcorrespondence}
    Let $\p \in \P$ be a program in $n$ variables and $\rho \in \R^{n+\omega}$.
    Then $\f_\p(\rho) = \rho'$ for some $\rho' \in \R^{n+\omega}$ iff there is $(F,B,\OO) \in \FF_\p$ such that $\rho \in B \cap \OO$. In this case, $F(\rho) = \rho' = \f_\p(\rho)$.
\end{lemma}
\begin{proof}
    The proof is by induction on the structure of $\p$.
    The base cases are trivial, because the symbolic semantics $\FF_{\p}$ are the singletons $\{(F,B,\OO)\}$ where $F$ is just $\f_\p$, $B$ is the entire space, and so is $\OO$, except in the case for observe, where it is the measurable set corresponding to the observed Boolean formula.
    \begin{itemize}[label=$\bullet$]
      \item If $\p = \pskip$ then $\f_{\p} = \id_{\statespace}$ and $\FF_{\p} = \{ (\id_{\statespace},\statespace,\statespace) \}$.
      \item If $\p = \assign$ then also $\FF_{\p} = \{ (\f_{\p},\statespace,\statespace) \}$.
      \item Idem for $\p = \psample$.
      \item If $\p = \observe$ then $\f_{\p}$ is the identity on $\interpret{\bexp}$ and $\f_\p(\rho)=\abort$ in case $\rho \in \compl{\interpret\bexp}$.
        Also, $(F,B,\OO)\in\FF_{\p}$ iff $(F,B,\OO)=(\id_{\statespace},\statespace,\interpret{\bexp})$.
        Thus, $\f_\p(\rho)=\rho$ iff $\rho\in\interpret\bexp=B\cap\OO$ and in this case, $F(\rho)=\rho=\f_\p(\rho)$.
    \end{itemize}
    In the inductive steps, the `if' and `only if' parts are proven separately. 
    \begin{itemize}[label=$\bullet$]
      \item First sequencing: let $\p = \seq {\p_1}{\p_2}$.
      \par$\underline{\Longrightarrow}:$\quad
            By definition, $\f_\p(\rho)=\f_{\p_2}(\f_{\p_1}(\rho))$.
            If $\f_\p(\rho)=\rho''$ for some $\rho'' \in \R^{n+\omega}$ then $\f_{\p_1}(\rho)=\rho'$ for some $\R^{n+\omega}$ and $\f_{\p_2}(\rho')=\rho''$.
            By IHs, then, there are $(F_1,B_1,\OO_1)\in\FF_{\p_1}$ and $(F_2,B_2,\OO_2)\in\FF_{\p_2}$ such that $\rho\in B_1\cap\OO_1$ and $\rho'\in B_2\cap\OO_2$.
            Moreover, $\f_{\p_1}(\rho)=F_1(\rho)$ and $\f_{\p_2}(\rho')=F_2(\rho')$.
            By definition of $\FF_{\p}$, $(F,B,\OO) := (F_2\circ F_1,B_1 \cap F_1^{-1}[B_2],\OO_1 \cap F_1^{-1}[\OO_2]) \in \FF_{\p}$.
            It can be straightforwardly verified that $\rho \in B \cap \OO$ and $F(\rho)=\f_\p(\rho)$.
      \par$\underline{\Longleftarrow}:$\quad
            Let $(F,B,\OO) \in \FF_{\p}$.
            By definition, it is composed of two $(F_i,B_i,\OO_i) \in \FF_{\p_i} (i=1,2)$ such that $F = F_2 \circ F_1$, $B = B_1 \cap F_1^{-1}[B_2]$, and $\OO = \OO_1 \cap F_1^{-1}[\OO_2]$.
            For $\rho\in B$, we have $\rho \in B_1$ and $F_1(\rho) \in B_2$.
            If $\rho \in \OO$ then $\rho \in B_1 \cap \OO_1$ so $\rho':=F_1(\rho)\overset{\textup{IH}}=\f_{\p_1}(\rho)$ for some $\rho'\in\R^{n+\omega}$.
            In addition, $\rho' \in B_2 \cap \OO_2$, so $\rho'':=F_2(\rho')\overset{\textup{IH}}=\f_{\p_2}(\rho')$ for some $\rho'' \in \R^{n+\omega}$.
            Hence, $\rho'' = \f_{\p_2}(\f_{\p_1}(\rho)) = \f_{\p}(\rho) $ and obviously $F(\rho)=\rho''$.
      \item Now consider $\p = \pif b{\p_1}{\p_2}$.
        \par$\underline{\Longrightarrow}:$\quad
            There are two cases to consider for the given $\rho$: $\rho\in\bexpset$ or $\rho\not\in\bexpset$.
            We prove only the second; the first is analogous.
            In that case, $\f_\p(\rho)=\f_{\p_2}(\rho)$.
            If $\f_{\p_2}(\rho)=\rho'$ for some $\rho'\in\R^{n+\omega}$ then by IH there is $(F,B,\OO)\in\FF_{\p_2}$ such that $\rho \in B\cap\OO$ and $F(\rho)=\rho'$.
            By definition, $(F,B\cap(\bexpcset),\OO)\in\FF_\p$ and, obviously, $\rho\in B\cap(\bexpcset)\cap\OO$.
        \par$\underline{\Longleftarrow}:$\quad
            Let $(F,B,\OO) \in \FF_{\p}$.
            Then either \rom{1} $B = B_1 \cap (\bexpset)$ and $(F,B_1,\OO) \in \FF_{\p_1}$, or \rom{2} $B = B_2 \cap (\bexpcset)$ and $(F,B_2,\OO) \in \FF_{\p_2}$.
            We prove only case \rom 1; case \rom 2 is analogous.
            Suppose $\rho \in B\cap\OO$ then, since $\rho\in B_1\cap\OO$, by IH,
            $F(\rho) = \f_{\p_1}(\rho) =\f_\p(\rho),$
            where the last equality uses $\rho \in \bexpset$.
      \item Finally, iteration: $\p = \while\bexp\q$.

        \noindent$\underline{\Longrightarrow}$:\quad
        If $\f_\p(\rho)=\rho'$ then there is $m\in\N$ such that $\f_\q^m(\rho) \in \bexpcset$ and for all $j<m$ in $\N$, $\f_\q^j(\rho) \in \bexpset$.
        By induction on $m$, we show that for every such $m$ and for every $\rho$, there is $(F,B,\OO) \in \F_{\bexp,\q}^m(\F_0)$, where denotes $\{(\id_\statespace,\bexpcset,\statespace)\}$, such that $\rho \in B\cap\OO$ and $F(\rho)=\f^m_\q(\rho)$.
        Verifying the base case $m=0$ with $(F,B,\OO) \in \F_0$ is routine.

        Suppose now, for the inductive step, that $\f_\q^{m+1}(\rho) \in \bexpcset$ and for all $j<m+1$ in $\N$, $\f_\q^j(\rho) \in \bexpset$.
        Then, putting $\rho':=\f_\q(\rho)$ this just says
        $\f_\q^m(\rho') \in \bexpcset$ and for all $j<m$ in $\N$, $\f_\q^j(\rho') \in \bexpset$.
        \begin{itemize}
          \item by the IH on $m$, now, there is $(F',B',\OO') \in \F_{\bexp,\q}^m(\F_0)$ such that $\f_\q(\rho) \in B' \cap \OO'$ and $F'(\f_\q(\rho))=\f^m_\q(\f_\q(\rho))$, and
          \item by the IH on the subterm $\q$, there is $(F_\q,B_\q,\OO_\q) \in \FF_\q$ such that $\rho \in B_\q \cap \OO_\q$ and $F_\q(\rho)=\f_\q(\rho)$, and
        \end{itemize}
        It is now mechanically verified that the element
        $$ (F,B,\OO) := (F'\circ F_\q, B_\q \cap F_\q^{-1}[B'] \cap (\bexpset), \OO_\q \cap F_\q^{-1}[\OO']) \in \F^{m+1}_{\bexp,\q}(\F_0) $$
        satisfies $\rho \in B \cap \OO$ and $F(\rho)=\f_\q^{m+1}(\rho)$.

        \noindent$\underline{\Longleftarrow}$:\quad
        Conversely, for $(F,B,\OO) \in \FF_\p$ we must have $m\in\N$ such that $(F,B,\OO)\in\F^m_{\bexp,\q}(\F_0)$, and we proceed again by induction on $m$.
        For $m=0$, if $\rho \in B\cap\OO$, then we know that $\rho \in \bexpcset$ and so, indeed $\min\{j \in \N \mid \f_\q^j(\rho)\in\bexpcset\}=0$, so $\f_\p(\rho)=\f_\q^0(\rho)=\rho=F(\rho)$.

        In the inductive step, let $(F,B,\OO) \in \F^{m+1}_{\bexp,\q}(\F_0)$. 
        Then, by definition of $\F_{\bexp,\q}$,
        $$ (F,B,\OO) = (F'\circ F_\q, B_\q \cap F_\q^{-1}[B'] \cap (\bexpset), \OO_\q \cap F_\q^{-1}[\OO']) $$
        for some $(F',B',\OO') \in \F^m_{\bexp,\q}(\F_0)$ and $(F_\q,B_\q,\OO_\q)\in\FF_\q$.
        For all $\rho \in B\cap\OO$, we also have $\rho \in B_\q \cap \OO_\q$ and so $\f_\q(\rho)=F_\q(\rho)$ by IH.
        Now $F_\q(\rho) \in B' \cap \OO'$ so by IH on $m$, we have $\f_\p(F_\q(\rho))=\f^m_\q(F_\q(\rho))$ and $m=\min\{j \in \N \mid \f_\q^j(F_\q(\rho))\in\bexpcset\}=0$.
        Since $F_\q(\rho)=\f_\q(\rho)$ and also $\rho \in \bexpset$, this means that
        $$ m+1 = \min\{j \mid \f_\q^j(\rho) \in \bexpcset\} $$
        and we conclude that $\f_\p(\rho)=\f^{m+1}_\q(\rho)=\f^m_\q(\f_\q(\rho))=F'(F_\q(\rho))=F(\rho)$.
        \end{itemize}
        This finishes induction over the structure of $\p$; we are done with the proof of \Cref{lem:ossymcorrespondence}.
  \end{proof}


  \begin{lemma}[Disjoint Path Probabilities] \label{lem:disjointpathconditions}
    If $(F,B,\OO), (F',B',\OO') \in \FF_{\p}$ are two distinct elements, then
    $B \cap B' = \emptyset$.
  \end{lemma}
  \begin{proof}
    The proof is by induction on the structure of $\p$.
    $\FF_{\p}$ is a singleton for all base cases of $\p$, so there there is nothing to prove there.
    \begin{itemize}
      \item Suppose $\p = \seq {\p_1}{\p_2}$ and let $(F,B,\OO)$ and $(F',B',\OO') \in \FF_{\p}$.
        By definition of $\FF_{\p}$, there are
        $(F_1,B_1,\OO_1), (F'_1,B'_1,\OO'_1) \in \FF_{\p_1}$ and
        $(F_2,B_2,\OO_2), (F'_2,B'_2,\OO'_2) \in \FF_{\p_2}$
        such that $F = F_2 \circ F_1$ and $F' = F'_2 \circ F'_1$; $B = B_1 \cap F_1^{-1}[B_2]$ and $B' = B'_1 \cap F'_1 {-1}[B'_2]$; and ${\OO} = {\OO}_1 \cap F_1^{-1}[\OO_2]$ and ${\OO}' = {\OO}'_1 \cap F'^{-1}_1[{\OO}'_2]$.
        \begin{enumerate}
          \item If $(F_1,B_1,\OO_1)$ and $(F'_1,B'_1,\OO'_1)$ are distinct, $B_1$ and $B'_1$ are disjoint by IH on $\p_1$.
            Thus $B$ and $B'$ are disjoint, as $B \subseteq B_1$ and $B' \subseteq B'_1$.
          \item Otherwise, $F_1=F_1'$ and $B_1=B'_1$, but $B_2$ and $B'_2$ must be disjoint by IH on $\p_2$, since $(F_2,B_2,\OO_2)$ and $(F_2,B_2,\OO_2)$ are distinct (otherwise $(F,B,\OO)$ and $(F',B',\OO')$ are equal).
            Then $F_1^{-1}[B_2]$ and $F'^{-1}_1[B'_2]$ are disjoint (since $F_1 = F'_1$) and so $B \cap B' = \emptyset$, since $B \subseteq F^{-1}_1[B_2]$ and $B' \subseteq F'^{-1}_1[B'_2]$.
        \end{enumerate}
        We have concluded that $B \cap B' = \emptyset$ in both cases, thus finishing this step.

      \item For $\p = \pif \bexp{\p_1}{\p_2}$, let $(F_1,B'_1,\OO_1), (F_2,B'_2,\OO_2) \in \FF_{\p}$ be distinct.
        There are two possibilities for $B'_1$:
        \begin{enumerate}
          \item $B'_1 = B_1 \cap (\bexpset)$ and $(F_1,B_1,\OO_1) \in \FF_{\p_1}$.
            Two cases also for $B'_2$:
            \begin{enumerate}
              \item $B'_2 = B_2 \cap (\bexpset)$ and $(F_2,B_2,\OO_2) \in \FF_{\p_1}$.
                In this case, $B_1$ and $B_2$ are disjoint because of the IH.
                Since $B'_1$ and $B'_2$ are subsets of these (respectively), the required disjointness follows.
              \item $B'_2 = B_2 \cap (\bexpcset)$ and $(F_2,B_2,\OO_2) \in \FF_{\p_2}$.
                Then $B'_1$ and $B'_2$ are disjoint because one is a subset of $\bexpset$ and the other of $\bexpcset$.
            \end{enumerate}
          \item $B'_1 = B_1 \cap (\bexpcset)$ and $(F_1,B_1,\OO_1) \in \FF_{\p_2}$.
            The rest of this case is a reasoning that is completely symmetrical to the above.
        \end{enumerate}
        In all cases, $B'_1 \cap B'_2 = \emptyset$, so the inductive step is finished.

      \item Finally, the case where $\p = \while{\bexp}{\q}$.
        Let $(F_1,B_1,\OO_1),(F_2,B_2,\OO_2) \in \FF_{\p}$.
        Then there are $m_1$ and $m_2$ such that
        $$(F_1,B_1,\OO_1)\in\F^{m_1}_{\bexp,\q}(\F_0)\quad\textup{and}\quad(F_2,B_2,\OO_2)\in\F^{m_2}_{\bexp,\q}(\F_0), $$
        where $\mathbb F_0 = \{(\id_{\statespace},\bexpcset,\statespace)\}$.
        Proceed by a case analysis:
        \begin{itemize}
          \item If $m_1=m_2=0$ then $(F_1,B_1,\OO_1)=(F_2,B_2,\OO_2)=(\id_\statespace,\bexpcset,\statespace)$, so there is nothing to prove.
          \item For the case $m_1=m_2>0$, we proceed by induction on $m:=m_1=m_2$:
            \begin{itemize}
              \item The base case is $m=1$. 
                There are two $(F_{\q,i},B_{\q,i},\OO_{\q,i}) \in \FF_\q$, $i=1,2$ such that $F_i=F_{\q,i}$, $B_i=B_{\q,i} \cap F_{\q,i}^{-1}[\bexpcset] \cap \bexpset$, and $\OO_i=\OO_{\q,i}$, by definition of $\F_{\bexp,\q}$.
                These two $(F_{\q,i},B_{\q,i},\OO_{\q,i})$, $i=1,2$ must be distinct, otherwise $(F_i,B_i,\OO_i)$, $i=1,2$ are equal.
                Then the sets $B_{\q,1}$ and $B_{\q,2}$ are disjoint by the IH for the subterm $\q$ of $\p$.
                A fortiori, $B_1$ and $B_2$ are disjoint.
              \item We now prove the statement for $m+1$.
                So let $(F_i,B_i,\OO_i)\in\F^{m+1}_{\bexp,\q}(\F_0)$, $i=1,2$.
                Then there are, for $i=1,2$, $(F_{\q,i},B_{\q,i},\OO_{\q,i}) \in \FF_\q$ and $(F'_i,B'_i,\OO'_i) \in \mathbb F^m_{\bexp,\q}(\F_0)$, such that
                \begin{enumerate}[label=(\roman*)]
                  \item $F_i = F'_i \circ F_{\q,i}$;
                  \item $B_i = (\bexpset) \cap B_{\q,i} \cap F_{\q,i}^{-1}[B'_i]$; and
                  \item $\OO_i = \OO_{\q,i} \cap F_{\q,i}^{-1}[\OO'_i]$.
                \end{enumerate}
                So we consider the following two cases:
                \begin{enumerate}
                  \item $(F_{\q,1},B_{\q,1},\OO_{\q,1})$ and $(F_{\q,2},B_{\q,2},\OO_{\q,2})$ are distinct.
                    By IH on $\q$, 
                    $$B_{\q,1} \cap B_{\q,2} = \emptyset$$
                  \item Otherwise, if $(F_{\q,1},B_{\q,1},\OO_{\q,1}) = (F_{\q,2},B_{\q,2},\OO_{\q,2})$ then $(F'_1,B'_1,\OO'_1)$ and $(F'_2,B'_2,\OO'_2)$ must be distinct.
                    But then by, by IH for $m$, we have $B'_1 \cap B'_2 = \emptyset$, and since $F_{\q,1} = F_{\q,2}$, we also have
                    $$ F_{\q,1}^{-1}[B'_1] \cap F_{\q,2}^{-1}[B'_2] =\emptyset $$
                \end{enumerate}
                In both cases $B_1$ and $B_2$ are disjoint, a fortiori.
            \end{itemize}
            Induction on $m$ has finished.
          \item The last possibility is $m_2 \neq m_1 > 0$, or, w.l.o.g., $m_2 > m_1 > 0$.

        \noindent We will use the following claim; write $\F_0=\{(\id_\statespace,\bexpcset,\statespace)\}$:
        \[ \forall m \in \N_{\geq 0} \,.\,
           \forall (F,B,\OO) \in \F^m_{\bexp,\q}(\F_0) \,.\,
         \exists (F_1,B_1,\OO_1), \dots, (F_m,B_m,\OO_m) \in \FF_\q \tag{*} \]
         s.t.
        \begin{enumerate}[label=(\roman*)]
          \item\label{item:whilecmp} $F = F_m \circ \dots \circ F_1$;
          \item\label{item:whileiter} $\forall j \in \{0,1,\dots,m-1\} \,.\, B \subseteq F_1^{-1}[\dots F_j^{-1}[(\bexpset) \cap B_{j+1}]\dots]$;
          \item\label{item:whileterm} $B \subseteq F^{-1}[\bexpcset]$.
        \end{enumerate}
        Here, \Cref{item:whileiter} for $j=0$ is just $B \subseteq (\bexpset) \cap B_1$.
        We prove all three items by induction on $m$.
        \begin{itemize}
          \item For $m=1$ we know $(F,B,\OO)=(F_\q,B_\q\cap(\bexpset)\cap F_\q^{-1}[\bexpcset],\OO_\q)$ for some $(F_\q,B_\q,\OO_\q)\in\FF_\q$.
            It is straightforwardly verified that by putting $(F_1,B_1,\OO_1):=(F_\q,B_\q,\OO_\q)$, \Cref{item:whilecmp,item:whileiter,item:whileterm} are satisfied.
          \item Now let $(F,B,\OO) \in \F^{m+1}_{\bexp,\q}(\F_0)$.
            Then there are $(F_\q,B_\q,\OO_\q) \in \FF_\q$ and $(F',B',\OO')\in\F^m_{\bexp,\q}(\F_0)$ such that $F=F'\circ F_\q$, $B=B_\q \cap (\bexpset) \cap F_\q^{-1}[B']$, and $\OO=\OO_\q \cap F_\q^{-1}[\OO']$.
            By the induction hypothesis, then, there are $(F'_1,B'_1,\OO'_1),\dots,(F'_m,B'_m,\OO'_m) \in \FF_\q$ such that
            \begin{enumerate}[label=(\roman*)]
              \item $F' = F'_m \circ \dots \circ F'_1$;
              \item $B' \subseteq {F'}^{-1}_1[\dots {F'}^{-1}_j[(\bexpset) \cap B'_{j+1}]\dots]$ for all $j \in \{0,1,\dots,m-1\}$; and
              \item $B' \subseteq {F'}^{-1}[\bexpcset]$.
            \end{enumerate}
            Put 
            $(F_1,B_1,\OO_1):=(F_\q,B_\q,\OO_\q)$ and 
            $(F_{j+1},B_{j+1},\OO_{j+1}):=(F'_j,B'_j,\OO'_j)$
            for $1 \leq j \leq m$.
            Then every $(F_j,B_j,\OO_j) \in \FF_\q$, and
            \begin{enumerate}[label=(\roman*)]
              \item $F = F' \circ F_\q = F'_m \circ \dots \circ F'_1 \circ F_\q = F_{m+1} \circ F_m \circ \dots \circ F_2 \circ F_1$;
              \item $B' \subseteq {F'}^{-1}_1[\dots {F'}^{-1}_j[(\bexpset) \cap B'_{j+1}]\dots]$ for all $j \in \{0,1,\dots,m-1\}$.
                And so, for all $j \in \{0,1,\dots,m-1\}$:
                $$\begin{array}{rl}
                  B \subseteq F_\q^{-1}[B'] \subseteq 
                  & F_1^{-1}[{F'}^{-1}_1[\dots {F'}^{-1}_j[(\bexpset)\cap B'_{j+1}]\dots]] \\
                  = & F_1^{-1}[F^{-1}_2[\ldots F^{-1}_{j+1}[(\bexpset)\cap B_{j+2}]\dots]]
                \end{array}$$
                This is to say that for all $j \in \{1,\dots,m\}$:
                $$ B \subseteq F_1^{-1}[\ldots F^{-1}_{j}[(\bexpset)\cap B_{j+1}]\dots] $$
                The case $j=0$ is trivially verified: $B \subseteq (\bexpset)\cap B_\q = (\bexpset)\cap B_1$.
              \item $B' \subseteq {F'}^{-1}[\bexpcset]$ so $B \subseteq F_\q^{-1}[{F'}^{-1}[\bexpcset]] = F^{-1}[\bexpcset]$.
            \end{enumerate}
        \end{itemize}
            We have verified all three properties for $(F,B,\OO)\in\F^{m+1}_{\bexp,\q}(\F_0)$, finished induction on $m$, and proved the claim $(*)$, which we will use now.

            Now let $(F,B,\OO) \in \F^{m_1}_{\bexp,\q}(\F_0)$ and $(F',B',\OO')\in\F^{m_2}_{\bexp,\q}(\F_0)$ (recall $m_2>m_1>0$).
            Then by the claim $(*)$ above, there are
            $$ \begin{array}{rl}
              (F_1,B_1,\OO_1),\dots,(F_{m_1},B_{m_1},\OO_{m_1}), &\\
              (F'_1,B'_1,\OO'_1),\dots,(F'_{m_2},B'_{m_2},\OO'_{m_2}) & \in \FF_\q
            \end{array}$$
            with the properties of \Cref{item:whilecmp,item:whileiter,item:whileterm}.
            \begin{itemize}
              \item Suppose first that 
                $$ (F_1,B_1,\OO_1)=(F'_1,B'_1,\OO'_1), \dots, (F_{m_1},B_{m_1},\OO_{m_1})=(F'_{m_1},B'_{m_1},\OO'_{m_1})$$
                By property \Cref{item:whileterm}, then, 
                $$B \supseteq F^{-1}[\bexpcset] = F_1^{-1}[\dots F_{m_1}^{-1}[\bexpcset]\dots]$$
                On the other hand, by property \Cref{item:whileiter}, since $m_1 \in \{0,\dots,m_2-1\}$,
                $$B' \supseteq {F'}^{-1}_1[\dots {F'}^{-1}_{m_1}[\bexpset]\dots] $$ 
                and since $F_j=F'_j$ for all $j \in \{1,\dots,m_1\}$, 
                $$ F_1^{-1}[\dots F_{m_1}^{-1}[\bexpcset]\dots] \cap {F'}^{-1}_1[\dots {F'}^{-1}_{m_1}[\bexpset]\dots] = \emptyset $$
              \item Otherwise let $M$ be the smallest integer in $\{1,\ldots,m_1\}$ such that we have $(F_M,B_M,\OO_M)\neq(F'_M,B'_M,\OO'_M)$.
                Then $$F_{M-1}\circ\dots\circ F_1 = F'_{M-1}\circ\dots\circ F'_1$$
                By property \Cref{item:whileiter},
                $$ B \supseteq F_1^{-1}[\dots F_{M-1}^{-1}[B_M]\dots ]
                \quad\textup{and}\quad
                B' \supseteq {F'}^{-1}_1[\dots {F'}^{-1}_{M-1}[B'_M]\dots ] $$ 
                By IH on $\q$, $B_M \cap B'_M = \emptyset$, so that now
                $$ F_1^{-1}[\dots F_{M-1}^{-1}[B_M]\dots ] \cap {F'}^{-1}_1[\dots {F'}^{-1}_{M-1}[B'_M]\dots ] = \emptyset $$
            \end{itemize}
            In both cases we conclude $B \cap B' = \emptyset$.
        \end{itemize}
    We have considered all possible values for $m_1$ and $m_2$.
  \end{itemize}
  Inductive proof of \Cref{lem:disjointpathconditions} is now finished.
  \end{proof}
\end{document}